\documentclass[preprint,12pt]{elsarticle}

\usepackage{fullpage}
\usepackage{amsmath,amssymb,amsthm} 
\usepackage{bm} 
\usepackage{mathrsfs} 
\usepackage[ruled,vlined]{algorithm2e} 
\usepackage{algpseudocode}
\usepackage{amssymb, colordvi} 
\usepackage{multirow} 
\usepackage{hyperref}
\usepackage{cleveref} 
\usepackage{graphicx}
\usepackage{blkarray} 
\usepackage{tikz} 
\usepackage{stmaryrd}
\usetikzlibrary{decorations}
\usepackage{centernot}
\usepackage{mathtools}

\usetikzlibrary{patterns}

\numberwithin{equation}{section} 
\newtheorem{theorem}{Theorem}
\numberwithin{theorem}{section}

\newtheorem{proposition}[theorem]{Proposition}
\newtheorem{example}[theorem]{Example} 
\newtheorem*{example*}{Example} 
\newtheorem{lemma}[theorem]{Lemma}
\newtheorem{corollary}[theorem]{Corollary}

\newtheorem{definition}[theorem]{Definition}

\newtheorem{notation}[theorem]{Notation}
\newtheorem*{remark*}{Remark}

\newcommand{\C}{\mathbb{C}}
\newcommand{\R}{\mathbb{R}}
\newcommand{\Z}{\mathbb{Z}}
\newcommand{\Q}{\mathbb{Q}}
\newcommand{\N}{\mathbb{Z}_{\geq 0}}
\newcommand{\PP}{\mathbb{P}}
\newcommand{\A}{\mathcal{A}}
\newcommand{\ideal}[1]{{\left\langle #1 \right\rangle}}
\newcommand{\conSet}[1]{{\left\llbracket #1 \right\rrbracket}}
\newcommand{\card}[1]{{\lvert #1 \rvert}}

\DeclareMathOperator{\Conv}{ConvexHull}

\DeclareMathOperator{\Div}{Div}
\DeclareMathOperator{\Spec}{Spec}
\DeclareMathOperator{\Hom}{Hom}
\DeclareMathOperator{\Cl}{Cl}
\DeclareMathOperator{\rank}{rank}
\DeclareMathOperator{\AffineSpan}{AffineSpan}
\DeclareMathOperator{\Proj}{Proj}
\let\div\relax
\DeclareMathOperator{\div}{div}
\DeclareMathOperator{\Aut}{Aut}
\DeclareMathOperator{\Faces}{Faces}

\DeclareMathOperator{\codim}{codim}

\journal{Journal of Algebra}

\begin{document}

\begin{frontmatter}

\title{
  Dimension Results for Extremal-Generic Polynomial Systems over Complete Toric
  Varieties
} 

\author[mb]{Matías Bender}
\ead{matias.bender@inria.fr}
 \affiliation[mb]{organization={Inria and CMAP, CNRS, École Polytechnique},
             addressline={Institut polytechnique de Paris},
             city={Palaiseau},
             country={France}}

\author[pjs]{Pierre-Jean Spaenlehauer}
\ead{pierre-jean.spaenlehauer@inria.fr}
 \affiliation[pjs]{organization={Université de Lorraine, Inria, CNRS},
             country={France}}

\begin{abstract}  
  We study polynomial systems with prescribed monomial
  supports in the Cox ring of a toric variety built from a complete
  polyhedral fan.
  We present combinatorial formulas for the dimension of their associated
  subvarieties under
  genericity assumptions on the coefficients of the polynomials.
  Using these formulas, we identify at which degrees generic systems in
  polytopal algebras form regular sequences.
  Our motivation comes from sparse elimination theory, where knowing
  the expected dimension of these subvarieties leads to specialized
  algorithms and to large speed-ups for solving sparse polynomial
  systems.
  As a special case, we classify the degrees at which regular sequences defined by weighted
  homogeneous polynomials can be found, answering an open question in the Gröbner
  bases literature.
  We also show that deciding whether a sparse system is generically a
  regular sequence in a polytopal algebra is hard from the point of view of
  theoretical computational complexity.
\end{abstract}


\end{frontmatter}

\section*{Introduction}

Toric varieties constructed from complete polyhedral fans form a large
class of compact varieties. They can be thought of as compactifications
of the affine space which have good properties for compactifying
affine varieties defined by polynomials with specific monomial
structures. This observation is the backbone of many results on the
algebra and geometry of sparse polynomial systems. The classical
choice is to construct a toric compactification by using the Minkowski
sum of the Newton polytopes of the polynomials. This leads to the
theory of mixed volumes and to the Bernstein–Khovanskii–Kushnirenko's
theorem~\cite{bernshtein_number_1975}.
In this paper, our main objective is to study properties of
alternative toric compactifications. A typical application and a
strong motivation for this work is the special case of weighted
polynomial systems: for these systems, the weighted projective space
is a natural toric compactification, but it is not always the one that
we get from the Newton polytopes of the polynomials. The weighted
projective space exhibits a typical difficulty: it is not smooth and
so we may need to deal with the intersection of non-Cartier divisors.

\smallskip

{\bf Main results.} We investigate the dimension of varieties defined
by generic homogeneous polynomials in the Cox ring of a toric variety
associated to a complete polyhedral fan. Such systems can be regarded
as homogenizations of affine polynomial systems with respect to a
toric compactification. The Cox ring has a canonical grading by the
divisor class group of the toric variety. We shall see that the
behavior of generic elements in the Cox ring whose degrees correspond
to non-Cartier divisor classes can be complicated: their associated
hypersurfaces may intersect non-properly torus orbits.  This
phenomenon may also occur for systems with prescribed monomial support
and generic coefficients.  This is featured in the following example:
\begin{example*}\label{ex:hyperelliptic}
Let us consider a complex genus-$g$ hyperelliptic curve given by an
equation of the form $y^2 - f(x) = 0$, where $f$ is a generic
univariate polynomial of degree $2g+1$. Homogenizations of
hyperelliptic equations are particularly important for cryptologists,
see. e.g.~\cite[Sec.~10.1.1]{galbraith_mathematics_nodate}. It is
well-known that the classical projective compactification obtained by
homogenizing this equation as $y^2 z^{2g-1} - z^{2g+1}f_i(x/z)=0$
introduces the point $(0:1:0) \in\PP^2$ at infinity. On the other
hand, the compactification with weight $g+1$ on the variable $y$
produces a curve with a point $(1:0:0)$ at infinity. The curve's
weighted homogeneous equation is $y^2 - z^{2g+2}f_i(x/z) = 0$. In both
cases, the point at infinity is a torus orbit: It is invariant by the
action of $\C^\times$.   
  \Cref{figure:hyperelliptic} illustrates the relation between the
  defining polynomials and the torus orbits of the toric
  compactifications in terms of polytopes.
  In both cases, we observe that the Newton polytope of the equation
  (dashed) is properly included in the polytope (solid gray)
  corresponding to the homogenization. The faces of the polytope in
  solid gray which do not intersect the dashed polytope correspond to
  the torus orbits which are contained in the closure of the curve
  with respect to the compactification. By the \emph{orbit-cone
    correspondence}, the vertex $(0, 5)$ on the left corresponds to
  the projective point $(0:1:0)$. On the right, the vertex $(6,0)$
  yields the weighted projective point $(1:0:0)$.
\end{example*}
The main thesis of this paper is that the torus
orbits that are generically included in toric compactifications of
algebraic varieties can be read off from the relationship between the
Newton polytope of the defining polynomials, the polyhedral fan
associated to the compactification, and the polytope associated to the
homogenization process.

\begin{figure}
  \centering
  \begin{tikzpicture}[scale=1]
    \draw[step=1cm,gray,thin, dotted] (-0.4,-0.1) grid (6.2,6.2);
\draw[thin, -latex] (0,0) -- (5.5,0);
\draw[thin, -latex] (0,0) -- (0,5.5);
\draw[thin, -latex] (0,0) -- (-1.2,-1.2);
  \filldraw[thick, fill= gray!30!white] (0,0) node[anchor=east] {$z^5$} -- 
  (1, 0) node[anchor=north] {$x z^4$} -- 
  (2, 0) node[anchor=north] {$x^2 z^3$} -- 
  (3, 0) node[anchor=north] {$x^3 z^2$} -- 
  (4, 0) node[anchor=north] {$x^4 z$} -- 
  (5, 0) node[anchor=north] {$x^5$} -- 
  (0, 5) node[anchor=east] {$y^5$} -- 
  (0, 2) node[anchor=east] {$y^2z^3$} -- 
  (0, 0);
  \filldraw[thick, pattern=north east lines] (0,0) -- 
  (5,0)  -- 
  (0, 2)  -- 
  (0, 0);
    \fill[fill=black] 
    (0, 0)  circle (3pt) -- 
    (1, 0)  circle (3pt)-- 
    (2, 0)  circle (3pt)-- 
    (3, 0)  circle (3pt)-- 
    (4, 0)  circle (3pt)-- 
    (5, 0)  circle (3pt)-- 
    (0, 2)  circle (3pt) -- 
  (0,0);
\end{tikzpicture}\quad\quad\quad
  \begin{tikzpicture}[scale=1]
    \draw[step=1cm,gray,thin, dotted] (-0.4,-0.1) grid (6.2,6.2);
\draw[thin, -latex] (0,0) -- (6.5,0);
\draw[thin, -latex] (0,0) -- (0,5.5);
\draw[thin, -latex] (0,0) -- (-0.25, -1);
  \filldraw[thick, fill= gray!30!white] (0,0) node[anchor=north] {$z^6$} -- 
  (1, 0) node[anchor=north] {$x z^5$} -- 
  (2, 0) node[anchor=north] {$x^2 z^4$} -- 
  (3, 0) node[anchor=north] {$x^3 z^3$} -- 
  (4, 0) node[anchor=north] {$x^4 z^2$} -- 
  (5, 0) node[anchor=north] {$x^5 z$} -- 
  (6, 0) node[anchor=north] {$x^6$} -- 
  (0, 2) node[anchor=east] {$y^2$} -- 
  (0, 0);
  \filldraw[thick, pattern=north east lines] (0,0) -- 
  (5,0)  -- 
  (0, 2)  -- 
  (0, 0);
    \fill[fill=black] 
    (0,0)   circle (3pt) -- 
    (1, 0)  circle (3pt)-- 
    (2, 0)  circle (3pt)-- 
    (3, 0)  circle (3pt)-- 
    (4, 0)  circle (3pt)-- 
    (5, 0)  circle (3pt)-- 
    (0, 2)  circle (3pt) -- 
  (0,0);
\end{tikzpicture}
  \caption{Two ways of homogenizing a imaginary genus-$2$ hyperelliptic
equation $y^2=f(x)$. On the left, the
projective homogenization in $\PP^2$; on the right, the homogenization with
respect to the weighted projective space, assigning weight $g+1$ on the variable
$y$.\label{figure:hyperelliptic}}
\end{figure}
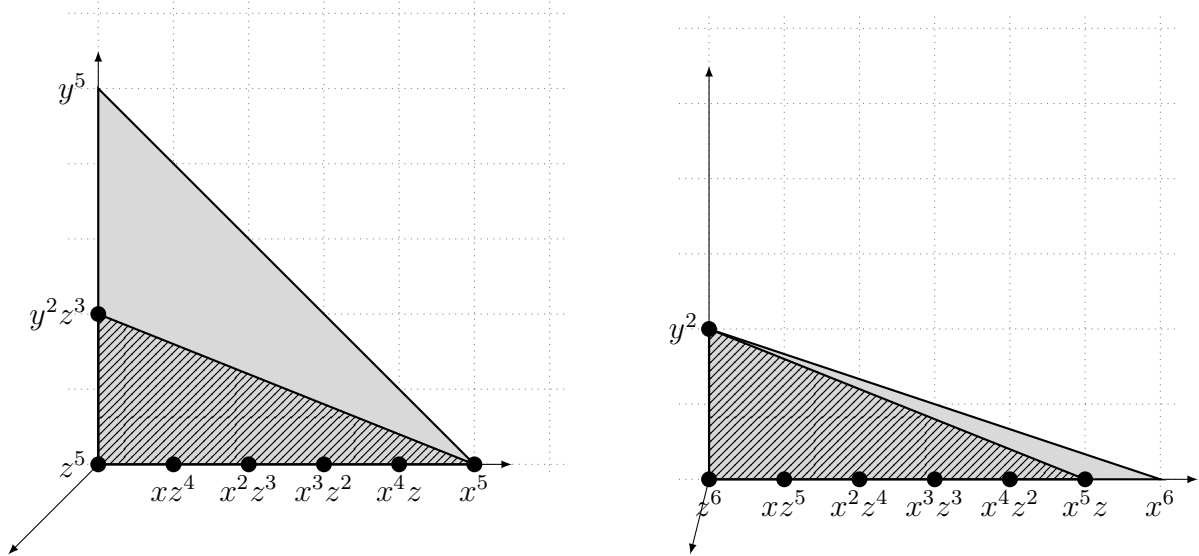

In Theorem~\ref{thm:dimensionGenericToricSystem},
we show that this relation can be controlled by a combinatorial criterion,
which provides a formula for the dimension of the subscheme defined by
generic homogeneous elements in the Cox ring in terms of their
degrees.
Our criterion extends to generic homogeneous polynomial systems with a
prescribed monomial support.
We also study the case of polytopal algebras constructed over rational
polytopes, which are $\N$-graded
homogeneous coordinate rings (not necessarily generated in degree one)
for projective toric varieties.
In this context, Theorem~\ref{thm:regSeqInPolytopalAlg} shows that our dimension results provide a
criterion to identify systems of polynomial equations which form
regular sequences under genericity assumptions on their coefficients.

An important case of study is the weighted projective space associated
to a sequence of weights $(w_0,\ldots, w_n)\in\Z_{> 0}^{n+1}$. For
$(d_1,\ldots, d_r)\in\Z_{> 0}^r$, Theorem~\ref{thm:weighthomregseq}
determines the dimension of the variety associated to a generic
sequence of weighted homogeneous polynomials
$f_1,\ldots, f_r\in \C[X_0,\ldots, X_n]$, where $\deg(X_i)=w_i$ and
$\deg(f_j)=d_j$. In particular, deciding when such systems are
regular sequences is an important question to estimate the complexity
of computing Gr\"obner bases~ \cite{faugere_complexity_2016}. This
question remained open before our work; we provide a complete
combinatorial answer to it; see \Cref{thm:weighthomregseq}.

Finally, we investigate whether our combinatorial criterion can be
checked easily from the point of view of computational complexity. In
Section~\ref{sec:complexity}, we
show that it is related to an NP-hard problem, which indicates
that automating the search for suitable toric compactifications might
not be an easy task, as previous works on multi-homogeneous systems
also suggest~\cite{malajovich2007computing}.

\smallskip

{\bf Applications and motivations.}
Our main motivation comes from symbolic computation and from the theory
of Gröbner bases. A fundamental problem in this area is to design
computational methods and tools for complexity analysis for solving
polynomial systems with specific monomial structures. A classical
approach to this question is to consider the toric compactification
associated to the normal fan of the Minkowski sum of the Newton
polytopes of the polynomials in the system.

In this context, solving methods based on numerical homotopies
\cite{morgan1987homotopy, huber1995polyhedral}, resultants
\cite{sturmfels1991sparse, canny1993efficient} and others
\cite{herrero2013affine, bender2022toric} have been developed and have
given birth to fruitful lines of work. During the last decades, new
methods related to Gröbner bases have been proposed, expanding the
symbolic toolbox for manipulating and solving sparse systems
\cite{faugere2014sparse, mourrain2014toric, bender_groebner_2019}. Our
long-term goal is to expand the toolbox to the case of toric
compactifications which do not necessarily correspond to the Newton
polytopes of the equations.
This line of work is currently very active in the context of homotopy
continuation algorithms: on the one hand, some works, e.g.
\cite{bihan2019criteria,Chen2019}, focused on constructing
polytopes with mixed volume equal to that of the
input Newton polytopes; on the other hand, when the input systems have
less solutions than the mixed volume of
their Newton polytopes, tools as toric deformations
\cite{burr2023numerical} and tropical geometry
\cite{helminck2022generic} were used in order to construct
more suitable compactifications.
In the context of Gr\"obner basis algorithms, there are complexity
bounds for the special case of weighted homogeneous systems
\cite{faugere_complexity_2016}.  Our new results on polytopal algebras
(Theorem~\ref{thm:regSeqInPolytopalAlg}) enhance the approach in
\cite{faugere2014sparse} and suggest that it might lead to more
general complexity bounds.

Deciding whether a compactification is good for a specific
application --- for instance deciding if the compactification of an
affine variety obtained via homogenization does not introduce
high-dimensional components, or deciding if the number of solutions
equals generically a number obtained combinatorially --- is also of
first importance. Unfortunately, this is a hard problem. For example,
in \cite{malajovich2007computing}, the authors study a simpler class
of compactifications restricted to multi-projective spaces and they
prove that finding the one that minimizes the multi-homogeneous Bézout
bound is an NP-hard problem.
Our results in Section~\ref{sec:complexity} complement this line of
work by showing that, in general, it is hard to decide whether
homogenizations create unwanted high-dimensional components.

\smallskip

{\bf Organization of the paper.} First, we introduce notations that we shall use
throughout this paper. In Section~\ref{sec:dimensionres}, we prove the main general dimension
result for extremal-generic systems in the Cox ring of a complete toric
variety.
Section~\ref{sec:completeinter} is devoted to the study of complete
intersections, showing in particular that the dimension
results can be strengthened in the more favorable case when the
degrees correspond to $\Q$-ample divisors.
In Section~\ref{sec:polytopalAlgebra}, we consider the special case of
systems in polytopal algebras, which correspond to homogeneous
coordinate rings for projective toric varieties.
We also identify the degrees at which regular sequences occur over weighted projective
spaces.
Finally, in Section~\ref{sec:complexity}, we end this paper by
investigating the theoretical complexity of computing generic
dimensions.

\section*{Notation}

Unless explicitly stated otherwise, we shall use the same (standard) notation as in \cite{cox2011toric}. In order to
avoid any problem related to the base field and positive characteristic, we
assume throughout this paper that every variety is defined over $\C$. Indeed,
most known results that we shall invoke are stated for toric varieties over
$\C$.

In what follows, we let $N$ be a Euclidean lattice, i.e. $N$ is a free abelian group
of finite rank $n$ endowed with a symmetric positive-definite bilinear form on
$N_\R := N\otimes_\Z\R$ defining a Euclidean norm. We let $M$ denote the dual lattice for the Euclidean
norm, and we set $M_\R:=M\otimes_\Z \R$. We use the notation $\C[M]$ to denote the Laurent polynomial ring with
exponents in $M$ and coefficients in $\C$. For $m\in M$, we let $\chi^m$ denote
the associated character in $\C[M]$. For an affine semigroup $\mathcal S \subset M$, we
let $\C[\mathcal  S]$ denote the corresponding subalgebra of $\C[M]$.

The letter $X$ denotes the $n$-dimensional normal toric variety
associated to the lattices $(N,M)$ and to a complete rational fan
$\Sigma$ over $N_\R$.
We let $\Sigma(1)$ denote the rays of $\Sigma$ and
$S = \C[x_\rho : \rho\in \Sigma(1)]$ denote the Cox ring of $X$ ---
called the \emph{total coordinate ring} of $X$
in~\cite[Ch.~5]{cox2011toric} --- with its natural grading by the
\emph{class group} $\Cl(X)$ of $X$~\cite[Def.~4.0.13]{cox2011toric}.
By slight abuse of notation, we say that a divisor class $\alpha\in\Cl(X)$ is
\emph{Cartier} if it is the class of a Cartier divisor; this corresponds to
elements in the Picard group of $X$ \cite[Thm.~6.0.20]{cox2011toric}. Similarly, we say that a divisor class is
\emph{effective} if it is the class of an effective divisor.
Given a cone $\sigma \in \Sigma$, we denote by $\sigma(1)$ the rays of
$\sigma$. For a ray $\rho\subset N_\R$, we let $u_\rho\in N$ denote its
primitive vector, i.e. the only element in $N$ satisfying $\Z_{\geq 0}
\cdot u_\rho =
\rho\cap N$.
We let
$B = \langle \prod_{\rho\notin \sigma(1)} x_\rho : \sigma\text{ cone
  of }\Sigma\rangle\subset S$ denote the irrelevant ideal.  Since
$\Sigma$ is complete, for each $\alpha\in\Cl(X)$, the $\C$-vector
space $S_{\alpha}$ of homogeneous elements of degree
$\alpha$ in $S$ is finite-dimensional
\cite[Prop.~4.3.8]{cox2011toric}.  For simplicity, we call
$T$-divisors the divisors on $X$ which are invariant under the action
of the torus~\cite[Exercise~4.1.1]{cox2011toric}. 
The class group $\Cl(X)$ is generated by the $T$-divisors
$\{D_\rho : \rho \in \Sigma(1)\}$, where $D_\rho$ is the divisor
defined by the closure of the $T$-orbit $O(\rho)$ associated to
$\rho \in \Sigma(1)$ \cite[Thm.~3.2.6]{cox2011toric}.

Given a finite-dimensional vector space $V \subset \C^n$, we say that
a property holds generically on $V$ if it holds on a dense subset
for the Zariski topology.
In this paper, $V$ is often a subspace of polynomial systems in
$S_{\alpha_1} \times \dots \times S_{\alpha_r}$, for degrees
$\alpha_1,\dots,\alpha_r \in\Cl(X)$. When $V$ is such a
finite-dimensional space of polynomials, we say that a property holds
for a generic system (in $V$) if it holds generically on $V$.
Given a degree $\alpha \in \Cl(X)$, a monomial set $\A \in S_\alpha$,
and a homogeneous polynomial $f \in S_\alpha$, we say that $f$ has
support $\A$ if all the monomials appearing in $f$ with nonzero coefficient
belong to $\A$ (monomials in $\A$ are allowed to appear with zero coefficient, which is important for the
proofs in Section~\ref{sec:polytopalAlgebra}).

Throughout this paper, all polytopes are supposed to be convex and bounded. 
We use the notation $\conSet{1, n}$ as a shorthand for
the interval of integers $\{1,\ldots, n\}$. The word \emph{family} is used to denote finite
multisets.

\section{Dimension results for extremal-generic
systems}\label{sec:dimensionres}

The aim of this section is to study the expected dimension of a
closed subscheme $Y$ associated to a generic system of homogeneous
polynomials of respective degrees
$\alpha_1,\ldots, \alpha_r \in \Cl(X)$ in the Cox ring of $X$.
To this end, we study the expected dimension of the intersection $Y$
with torus orbits.
Our study goes beyond generic systems and can be applied to a broader
class of systems that we call extremal-generic
(\Cref{def:extremal-generic}).

The main property that we exploit is that toric varieties associated
to complete fans can be decomposed as a disjoint union of finitely many tori.
This decomposition is
called the \emph{orbit-cone correspondence}, and it describes the
toric variety $X$ as the union of the orbits of the points in $X$
under the action of the torus. Moreover, there is a bijection between
the cones of $\Sigma$ and these torus orbits, see
e.g.~\cite[Chapter~3]{cox2011toric}.
In order to compute the dimension of the subvariety $Y$, we study how
such a subvariety intersects with torus orbits $O(\sigma)\subset X$
associated to cones $\sigma$ in $\Sigma$. Then, the dimension of $Y$
will correspond to the maximum among the dimensions of these
intersections.

\subsection{Hypersurfaces}

In this section, we will focus on hypersurfaces $Y$ defined by
$f \in S_\alpha$, for some fixed $\alpha \in \Cl(X)$. Given a cone
$\sigma$ in $\Sigma$, we shall see that one of the following three
situations can happen, and this can be discriminated by investigating
combinatorial properties of $\sigma$:

\begin{itemize}
  \item (incompatibility case) All hypersurfaces of $X$ associated to
    polynomials in $S_{\alpha}$ contain the torus orbit $O(\sigma)$;
  \item (non-essential case) The orbit $O(\sigma)$ does not intersect any hypersurface
    associated to a polynomial in $S_{\alpha}$;
  \item (essential case) a generic polynomial in
    $S_{\alpha}$ defines a hypersurface in $O(\sigma)$.
\end{itemize}

We will compute the intersection of closed subschemes in $X$ with
a torus orbit $O(\sigma)$ in a specific ambient affine scheme given
by $U_\sigma = \cup_{\tau\subset\sigma} O(\tau)\subset X$ where $\tau$
ranges over the subcones of $\sigma$. Then $U_\sigma$ is an affine
toric variety isomorphic to $\Spec(\C[\sigma^\vee \cap M])$
\cite[Thm.~1.2.18]{cox2011toric}.
An enlightening way of thinking about the coordinate ring of
$U_\sigma$ is by using the isomorphism between
$\C[\sigma^\vee \cap M]$ and the degree-zero part
$(S_{x^{\hat{\sigma}}})_0$ of the localization of $S$ at
$x^{\hat{\sigma}} :=
\prod_{\substack{\rho\in\Sigma(1)\\\rho\notin\sigma(1)}}x_\rho$.
This isomorphism is given by
\begin{align}\label{eq:isom_local}
  \begin{array}{l c c l}
\pi_\sigma^* : & \C[\sigma^\vee \cap M] & \xrightarrow{\sim} &
(S_{x^{\hat{\sigma}}})_0, \\
& \chi^m & \mapsto & \pi_\sigma^*(\chi^m) = \displaystyle\prod_{\rho \in \Sigma(1)} x_\rho^{\langle m ,
  u_\rho \rangle}.
  \end{array}
\end{align}
For more details, see \cite[Thm~5.1.11]{cox2011toric}.
Let $V(\sigma)$ be the closure of the torus orbit $O(\sigma)$ in $X$.
According to the orbit-cone correspondence,
$V(\sigma) \cap U_\sigma = O(\sigma)$.
Hence, $O(\sigma)$ is closed in $U_\sigma$, and
by~\cite[Thm~3.2.6]{cox2011toric}, the ideal in $\C[\sigma^\vee
\cap M]$ defining $O(\sigma)$ is
\begin{align*}
J_\sigma = \ideal{\chi^{m} : m \in \sigma^\vee \cap M \text{ but } m \not\in
  \sigma^\perp}.
\end{align*}

A useful remark is that the exponent vectors $m$ of the monomials
$\chi^m$ in $\C[\sigma^\vee \cap M]$ which do not belong to $J_\sigma$
correspond to the points in $\sigma^\perp \cap M$, so they lie on a
face of $\sigma^\vee$, and we have that
$$\C[\sigma^\vee \cap M] / J_\sigma \cong \C[\sigma^\perp \cap M].$$

The isomorphism from Equation~\eqref{eq:isom_local} provides us with a
canonical map which sends an ideal in $S$ to an ideal in
$\C[\sigma^\vee \cap M]$ by considering the zero-degree part of its
localization at $x^{\hat{\sigma}}$. It would be convenient if this map
sends principal ideals to principal ideals, since in this case this
map could be thought of as a restriction morphism. For instance, this
happens in the case of the projective plane where a polynomial $f$ of
degree $d$ in $\C[X,Y]$ can be thought of as the restriction to the
affine space $\{(x:y:z)\mid z\neq 0\}$ of the function $f(X/Z, Y/Z)$
on $\mathbb P^2$. Unfortunately, this nice property is not satisfied
in general, as illustrated by the following example.

\begin{example} \label{ex:nonPrincipal} Set $M = \mathbb Z^2$ and let
  $\Sigma$ be the complete fan in $\mathbb R^2$ with rays
  $\rho_1=\mathbb R_{\geq 0} (2,3)$,
  $\rho_2 = \mathbb R_{\geq 0}(-1, 0)$,
  $\rho_3=\mathbb R_{\geq 0} (0,-1)$. The associated toric variety is
  the weighted projective space $\PP(1,2,3)$, with Class group
  $\Cl(\PP(1,2,3))\simeq\mathbb Z$ and Cox ring $S=\C[X,Y,Z]$. The
  grading $S$ is given by $\deg(X^i Y^j Z^k) = i+2j +3k$. Then the
  polynomials of degree $2$ in $S$ form a $2$-dimensional
  vector space generated by $X^2$ and $Y$. Let $\sigma$ be the
  $2$-dimensional cone generated by $\rho_1$ and $\rho_2$. The
  intersection
  $\langle Y\rangle\cap (S_{Z})_0 = \langle XY/Z, Y^3/Z^2\rangle$ of
  the localized principal ideal $\langle Y\rangle \subset S_{Z}$ with
  the (weighted) degree-zero subring is not principal in the ring
  $(S_{Z})_0$.
\end{example}

However, the following two lemmas show that the situation is in fact
better when we focus on ideals belonging to
$\C[\sigma^\perp \cap M]$: principality is preserved
under a non-degeneracy condition. This condition, which is not
satisfied by the example above, asks for the existence of an effective
$T$-divisor with valuation zero at the rays of the cone.
When this condition is not satisfied, the restriction of the ideal to
$\C[\sigma^\perp \cap M]$ corresponds to the zero ideal.

\begin{lemma}\label{lem:compatible}
 Let $\sigma$ be a cone in $\Sigma$ and
  $D=\sum_{\rho\in\Sigma(1)} a_{\rho}D_{\rho}$ be an effective
  $T$-divisor such that $a_\rho = 0$ for $\rho\in \sigma(1)$. Let
  $\alpha:=[D]\in\Cl(X)$ be the class of $D$.
 Then for any $f\in S_\alpha$, the set $(f\cdot (S_{x^{\hat\sigma}})_{-\alpha})$ is a principal ideal in $(S_{x^{\hat\sigma}})_0$
 generated by $\frac{f}{m}$ where
 $ m  := \prod_{\rho\in(\Sigma(1)\setminus\sigma(1))}x_\rho^{a_\rho}.$
\end{lemma}

\begin{proof}
  Consider any polynomial $g$ of degree zero belonging to the principal ideal
  generated by $f$ in the localized ring $S_{x^{\hat\sigma}}$. Then, we can write $g$ as
  $\hat{g} \cdot f$, where
  $\hat{g} \in (S_{x^{\hat\sigma}})_{-\alpha}$.
  As $\alpha = [D]$ and $a_\rho = 0$ for $\rho\in \sigma(1)$, we have
  that the monomial
  $m
  \in (S_{x^{\hat\sigma}})_{\alpha}$ is invertible in
  $S_{x^{\hat\sigma}}$.
  Hence, $g = \frac{m}{\hat{g}} \cdot \frac{f}{m}$ belongs to the principal ideal generated by $\frac{f}{m} \in (S_{x^{\hat\sigma}})_{0}$.
\end{proof}

When the requirements of
\Cref{lem:compatible} are not satisfied, the following lemma shows that all
hypersurfaces in $X$
associated to polynomials in $S_{\alpha}$ contain $O(\sigma)$. This corresponds to the \emph{incompatibility case} described at the
beginning of the section.

\begin{lemma}\label{lem:noncompatible}
  Let $\alpha\in\Cl(X)$ be a divisor class and $\sigma$ be a cone in $\Sigma$.  If there does
  not exist an effective $T$-divisor of the form
  $D=\sum_{\rho\in(\Sigma(1)\setminus\sigma(1))} a_{\rho}D_{\rho}$ in
  the class $\alpha$, then for any $f \in S_{\alpha}$ the ideal
  $(f\cdot (S_{x^{\hat\sigma}})_{-\alpha})\subset (S_{x^{\hat\sigma}})_0$ is contained in $\pi_\sigma^*(J_\sigma)$.
\end{lemma}
\begin{proof}
  We prove the lemma by contraposition. Let $f\in S_\alpha$ be such that
  $(f\cdot (S_{x^{\hat\sigma}})_{-\alpha}) \not\subset \pi_\sigma^*(J_\sigma)$.
  Then there exists a monomial $m_1$ of $f$ such that $(m_1\cdot
  (S_{x^{\hat\sigma}})_{-\alpha}) \not\subset \pi_\sigma^*(J_\sigma)$.
  Therefore, there exists a monomial $m_2\in (S_{x^{\hat\sigma}})_{-\alpha}$
  such that $m_1m_2 \not\in \pi_\sigma^*(J_\sigma)$.
  This means that
  the principal $T$-divisor associated to the character $\chi^{m_1 \, m_2}$ has valuation zero
  at all the rays of $\sigma$.
  Since $\div(m_2)=\sum_{\tau\in\Sigma(1)} \lambda_\tau D_\tau$ is such that
  $\lambda_\tau
  \geq 0$ for all $\tau\in\sigma(1)$,
  $\div(m_1)=\sum_{\tau\in\Sigma(1)}\mu_\tau D_\tau$ must satisfy $\mu_\tau=0$
  for all $\tau\in\sigma(1)$. Moreover, $\div(m_1)$ is in the class $\alpha$, so
  it satisfies the desired properties. 
\end{proof}

Our next step is to understand the intersection of $O(\sigma)$ and of a
hypersurface defined by an element $f \in S_\alpha$ when the requirements of
\Cref{lem:compatible} are satisfied. As we
will see, this intersection is either empty (\emph{non-essential
  case}) or it defines a hypersurface in $O(\sigma)$ (\emph{essential
  case}).
In order to study this intersection, we consider the rings
$\C[\sigma^\vee \cap M]$ and $\C[\sigma^\perp \cap M]$ and we give
a concrete construction of the inverse of the map in
Eq.~\eqref{eq:isom_local}.
For a cone $\sigma \in \Sigma$ and an
effective Weil divisor $D$ of the form
$\sum_{\rho \in \Sigma(1)\setminus\sigma(1)} v_\rho \, D_\rho$ ---
i.e. $D$ has valuation zero at all rays in $\sigma(1)$ --- we define
a polytope together with its set of lattice points:
\begin{align}
  \label{eq:defPolytopeDivisor}
  \begin{array}{rcl}
    P_{D,\sigma} &:=& \left\{
 \begin{array}{l | l  l}
 m \in M_\R  & \langle m, u_\rho \rangle \geq - v_\rho & \text{for all } \rho \in \Sigma(1),  \\
 & \langle m, u_\rho \rangle = 0  & \text{for all } \rho \in \sigma(1)
 \end{array}
\right\}\subset \sigma^\perp.\\
A_{D,\sigma} &:=& (P_{D,\sigma}\cap M) \subset (M\cap\sigma^\perp).
  \end{array}
\end{align}

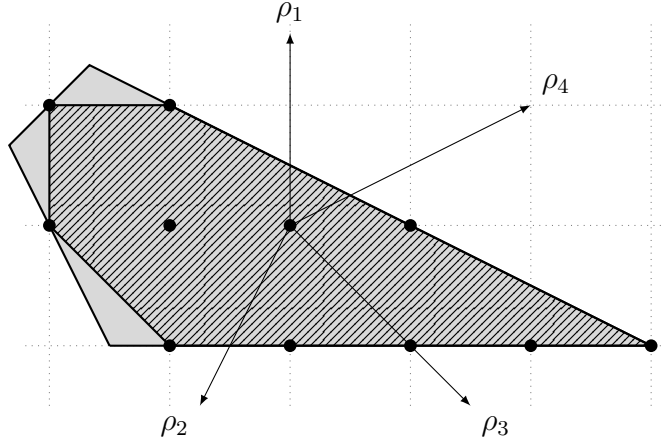
\begin{figure}
  \centering
  \begin{tikzpicture}[scale=1.6]
\draw[step=1cm,gray,thin, dotted] (-2.2,-1.5) grid (3.2,1.7);
    \filldraw[thick, fill= gray!30!white] (-3/2, -1) -- (-7/3, 2/3) -- 
    (-5/3, 4/3) --  (3, -1) -- (-3/2, -1);
\draw[thin, -latex] (0,0) -- (3/2,-3/2) node[below right] {$\rho_3$};
\draw[thin, -latex] (0,0) -- (0,1.6) node[above] {$\rho_1$};
\draw[thin, -latex] (0,0) -- (2, 1) node[above right] {$\rho_4$};
\draw[thin, -latex] (0,0) -- (-3/4,-3/2)  node[below left] {$\rho_2$};
\fill[fill=black] (-2,0) circle (1.5pt);
\fill[fill=black] (-2,1) circle (1.5pt);
\fill[fill=black] (-1,-1) circle (1.5pt);
\fill[fill=black] (-1,0) circle (1.5pt);
\fill[fill=black] (-1,1) circle (1.5pt);
\fill[fill=black] (0,0) circle (1.5pt);
\fill[fill=black] (0,-1) circle (1.5pt);
\fill[fill=black] (1,0) circle (1.5pt);
\fill[fill=black] (1,-1) circle (1.5pt);
\fill[fill=black] (2,-1) circle (1.5pt);
\fill[fill=black] (3,-1) circle (1.5pt);
  \filldraw[thick, pattern=north east lines] 
    (-2,1) -- (-1,1) -- (3,-1) -- (-1,-1) -- (-2, 0) -- (-2, 1); 
\end{tikzpicture}
  \caption{This picture represents the polytope $P_{D, \{0\}}$ and the set
  $A_{D,\{0\}}$ for the effective divisor
  $D=D_{\rho_1}+D_{\rho_2}+3\,D_{\rho_3}+4\,D_{\rho_4}$
  on the toric variety $X$ defined by the $2$-dimensional complete fan whose
  rays are $\rho_1=\R_>\cdot(0,1), \rho_2=\R_>\cdot(-1, -2), \rho_3=\R_>\cdot(1, -1),
  \rho_4=\R_>\cdot(2,1)$. This example shows that the convex hull of
  $A_{D,\{0\}}$ need not be equal to $P_{D,\{0\}}$.\label{fig:examplePA}}
\end{figure}

Since $\Sigma$ is complete, $P_{D,\sigma}$ is a polytope and the set
$A_{D,\sigma}$ is finite (see e.g.~\Cref{fig:examplePA}). Note that $P_{D,\{0\}}$ is the classical polytope
associated to a divisor, see e.g.~\cite[Eq.~4.3.2]{cox2011toric}.
An important remark is that if $D$ and $D'$ are linearly equivalent $T$-divisors with
valuation zero at all rays of $\sigma$, then
there exists a unique $m\in M\cap \sigma^\perp$
such that $D'= D+\div(\chi^m)$, and therefore $A_{D', \sigma} = A_{D, \sigma} +
m$.
This observation leads to the following proposition which is essentially a
restatement of \cite[Prop.~4.3.3]{cox2011toric}, see also \cite[Sec.~3.4]{fulton1993introduction}.

\begin{proposition}(\cite[Prop.~4.3.3]{cox2011toric})\label{prop:bij_pol_Sa}
  Let $D$ be a $T$-divisor of class $\alpha=[D]\in\Cl(X)$. Then $D$ defines a bijection
  between the monomials in $S_\alpha$ and $A_{D,\{0\}}$ by sending $\prod_{\rho \in \Sigma(1)}
  x_\rho^{a_\rho}\in S_\alpha$ to the uniquely defined $m\in A_{D,\{0\}}$ such
  that $\sum_{\rho \in \Sigma(1)} a_\rho D_\rho =  D +\div(\chi^m)$.
\end{proposition}

\begin{proof}
Write $D = \sum_{\rho \in \Sigma(1)} v_\rho \, D_\rho$. Since $\Sigma$ is
  complete, it has no torus factor, hence there exists a unique $m\in M$ such that $\sum a_\rho D_\rho =
  D +\div(\chi^m) = D +\sum \langle m, u_\rho\rangle
  D_\rho$~\cite[Thm.~4.1.3]{cox2011toric}.
  As the coefficients $a_\rho$ are nonnegative, we must have $\langle m,
  u_\rho\rangle\geq -v_\rho$, which shows that $m\in A_{D, \{0\}}$. This map is
  injective since $\div(\chi^{m_1})=\div(\chi^{m_2})$ if and only if $m_1 =
  m_2$. Surjectivity comes from the fact that for any $m\in
  A_{D,\{0\}}$, the divisor $D +\div(\chi^m) = \sum_{\rho\in\Sigma(1)} a_\rho
  D_\rho $ is
  effective and belongs to the class $\alpha$, so it is the image of
  $\prod_\rho
  x_\rho^{a_\rho}\in S_\alpha$.
\end{proof}

\begin{proposition}\label{prop:map_compatible}
 Let $\sigma$, $D$, $\alpha$ be as in Lemma~\ref{lem:compatible}. We consider
  the $\C$-linear map $S_\alpha\rightarrow
  \C[\sigma^\perp\cap M]$ which maps a monomial in $S_\alpha$ to its image via
  the bijection in Proposition~\ref{prop:bij_pol_Sa} if it belongs to $A_{D,
    \sigma}$, or to zero if it does not.
    This map sends monomials to monomials, and every point in $A_{D,\sigma}$ is
    the image of a monomial in $S_\alpha$.
\end{proposition}

\begin{proof}
    Notice that $A_{D,\{0\}}\subset \sigma^\vee\cap M$ since $D$ has valuation
  zero at all the rays in $\sigma$. This inclusion represents the map
  $S_\alpha\rightarrow (S_{x^{\hat\sigma}})_0$ on monomials.
  Finally, we check that $A_{D,\{0\}}\cap \sigma^\perp = A_{D,\sigma}$, which
concludes the proof. 
\end{proof}

For $f$ in $S_\alpha$, we let $f^\sigma$ denote the image of $f$ via
the map defined in Proposition~\ref{prop:map_compatible}, where for
simplicity we omit the dependency on the divisor $D$. We emphasize
that choosing another divisor $D$ amounts to multiplying $f^\sigma$ by
a monomial in $\C[\sigma^\perp\cap M]$, so the ideal
$\langle f^\sigma\rangle\subset\C[\sigma^\perp\cap M] = \C[\sigma^\vee
\cap M] / J_\sigma$ is independent of the choice of $D$. For
convenience, we set $f^\sigma = 0$ if we are in the case of
Lemma~\ref{lem:noncompatible}, i.e. when there does not exist any
effective $T$-divisor in $\alpha$ with valuation $0$ at all rays of
$\sigma$, and hence $f \in \pi_\sigma^*(J_\sigma)$.

\begin{lemma}\label{lem:disting_cases}
  Let $f \in S$ be a homogeneous polynomial and $Y$ be the subscheme of $X$
  associated to $\ideal{f}$. Let $\sigma$ be a cone in $\Sigma$. Then:
  \begin{itemize}
  \item (non-essential case) If $f^\sigma = 0$, then
    $O(\sigma) \subset Y$;
  \item (incompatibility case) If $f^\sigma$ involves exactly one monomial with
    nonzero coefficient,
    then $Y \cap O(\sigma)$ is empty;
  \item (essential case) If $f^\sigma$ involves more
    than one monomial with nonzero coefficient, then $Y \cap O(\sigma)$ is a hypersurface in
    $O(\sigma)$.
  \end{itemize}
\end{lemma}

\begin{proof}
  This is a direct consequence of
  Proposition~\ref{prop:map_compatible} and Krull's principal ideal
  theorem.
\end{proof}

Lemma~\ref{lem:disting_cases} states that distinguishing between the
cases depends only on the support of $f$, and therefore
generically it is determined by $A_{D,\sigma}$.

\subsection{General case}

In this section, we focus on the expected dimension of
$Y \cap O(\sigma)$, where $Y$ is defined by a sequence of
\emph{extremal-generic} polynomials. Before defining this notion of
genericity, we start by noting that, given a homogeneous
ideal in $S$ defining a closed subscheme $Y$, there is a simple
description of the ideal in $\C[\sigma^\perp \cap M]$ defining
$Y \cap O(\sigma)$.

\begin{corollary}\label{coro:scheme_intersection}
Let $f_1,\ldots, f_r\in S$ be homogeneous elements, and let $Y\subset X$ denote
the closed subscheme associated to the ideal $\langle f_1,\ldots, f_r\rangle$. Then
the ideal in $\C[\sigma^\perp\cap M]$ defining $Y\cap O(\sigma)$ is $\langle
f_1^\sigma,\ldots,f_r^\sigma\rangle$.
\end{corollary}

\begin{proof}
  First, we notice that if $H_1,\ldots, H_r\subset X$ denote the hypersurfaces
  associated to $f_1,\dots, f_r$, then $(H_1\cap\dots\cap H_r)\cap O(\sigma) =
  (H_1\cap O(\sigma))\cap \dots\cap (H_r\cap O(\sigma))$. Hence, it is
  sufficient to prove the result for $r=1$, which is a direct consequence of Proposition~\ref{prop:map_compatible}.
\end{proof}

Following classical results on sparse polynomial systems, e.g.
\cite{canny1991optimal}, we
will only assume that the genericity conditions are satisfied for the
coefficients which correspond to the vertices of the convex hull of the
monomial support. We will call these systems
\emph{extremal-generic}. We will first define this notion for
polynomials in $\C[M]$ and show how we can use it to predict the
codimension of their zero locus over the torus. Later,
in \Cref{def:Coxringextremal}, we extend this definition for
polynomials in the Cox ring of $X$. We will use this notion to
study the generic codimension of $Y$.

\begin{definition} \label{def:extremal-generic}
Given a finite subset $A\subset M$, the vertices of the convex hull of $A$ in
  $M_\R$ are called the \emph{extremal points} of $A$.
  Given a Laurent polynomial $f$ in $\C[M]$ with support $A$, we call
  \emph{extremal coefficients} of $f$ the coefficients of the monomials
  corresponding to the extremal points of $A$. Given $\bm A=(A_1,\ldots, A_r)$,
  $A_i\subset M$,
  with respective extremal points $(V_1,\ldots, V_r)$, we say that a
property holds for an \emph{extremal-generic} system if for any values
$(\alpha_u^{(i)})\in\C^{\sum_i\card{A_i\setminus V_i}}$ of
the non-extremal coefficients, there is a Zariski-dense subset $W\subset
\C^{\sum_i\card{V_i}}$ such that systems with non-extremal coefficients equal
  to $(\alpha_u^{(i)})$ and extremal coefficients in $W$ satisfy this property.
\end{definition}

The notion of extremal-genericity is a generalization of the classical
genericity.
Given supports $\bm A=(A_1,\ldots, A_k)$, we consider the space of all
systems with support $\bm A$ in $\C[M]$. This space is isomorphic to
$\C^{\sum_i \card{A_i}}$ endowed with the Zariski topology. Separating
the coefficients corresponding to the vertices of the convex hulls of
$A_1,\ldots, A_k$ in $M_\R$, we can write this space as a product
$\C^{n_1}\times \C^{n_2}$, where the first factor corresponds to the
coefficients of the extremal coefficients. The next lemma shows that a property that holds for extremal-generic
systems also holds for generic systems (i.e. systems for which the
genericity assumption holds on all coefficients).

\begin{lemma}\label{lem:genericity}
  Let $W_1\times W_2$ be a topological space such that $\{x\}\times W_2$ is
  closed for all $x\in W_1$. For each $x\in W_1$, let
  $\{x\}\times O_x$ be a dense subset of $\{x\}\times W_2$ for the induced topology. Then
  $\cup_{x\in W_1} (\{x\}\times O_x)$ is dense in $W_1\times W_2$.
\end{lemma}

\begin{proof}
  Assume by contradiction that $\cup_{x\in W_1} (\{x\}\times O_x)$ is not dense,
  then it is contained in a proper closed subset $V\subset W_1\times W_2$. Therefore, there exists
  $y\in W_1$ such that $V\cap (\{y\}\times W_2)\neq (\{y\}\times W_2)$. Since $V\cap
  (\{y\}\times W_2)$ is closed, we get that $\{y\}\times O_y$ is not
  dense, which contradicts our assumption.
  \end{proof}

Our next goal is to understand how extremal-generic systems with given
monomial support intersect over the torus.
Here, our analysis relies on classical tools related to the theory of
sparse resultants, see e.g.~\cite[Sec.~4.3]{sturmfels2002solving}. A
central concept is the notion of \emph{essential family}, which
discriminates whether the monomial supports of the system lead
generically to an empty intersection, or if generic hypersurfaces
intersect properly.

\begin{definition}[Essential family] \label{def:essentialFam}
A family of finite subsets $A_1,\ldots, A_{r} \subset M$ is
\emph{essential} if $\dim_\R\left(\AffineSpan_\R(\sum_{i\in E} A_i)\right) \geq
  \lvert E\rvert$ for every subset $E\subseteq\conSet{1, r}$, where
$\AffineSpan_\R(\sum_{i\in E} A_i)$ denotes the smallest affine subspace in $M_\R$
containing $\sum_{i\in E} A_i$.
\end{definition}

The following lemma is a variant of \cite[Thm.~1.1]{sturmfels1994newton}, which
will be a useful technical tool to compute the expected dimension of extremal-generic systems over the torus.

\begin{lemma}\label{lem:technical_lemma_sturmfels}
  Let $\bm A=(A_1,\ldots, A_r)$ be a family of $r$ nonempty finite subsets of
  $M$. For $i\in\conSet{1, r}$, let $V_i\subseteq A_i$ be the subset
  of vertices of the convex hull of $A_i$. Let
  $f_1,\dots,f_r \in \C[M]$ be polynomials with
  supports $A_1,\ldots, A_r$ and indeterminate vertex coefficients:
  $$f_i = \sum_{u\in
  V_i}\mathfrak v_u^{(i)} \chi^u + \sum_{u\in A_i\setminus V_i}c^{(i)}_u \chi^u,$$
  for some fixed $c\in \C^{\sum_i\lvert A_i\setminus V_i\rvert}$. Let $Z\subseteq \mathbb \C^{\lvert
  V_1\rvert}\times\cdots\times\mathbb \C^{\lvert
  V_r\rvert}$ denote the set of such systems which have complex solutions in
  $(\C^*)^n$. The codimension of its Zariski closure $\overline Z$ is
  $$\max_{E\subseteq\conSet{1,
  r}}\left(\card{E}-\dim_\R\left(\AffineSpan_\R\left(\sum_{i\in E} A_i\right)
  \right)\right).$$
\end{lemma}

\begin{proof}
  The proof is the same as the proof of
  \cite[Thm.~1.1]{sturmfels1994newton}, but replacing Bernstein's
  theorem by its version with generic extremal coefficients
  \cite{canny1991optimal}.
\end{proof}

The dimension of an extremal-generic sparse system is determined by
the essentiality of its exponent sets.

\begin{proposition}[Extremal-generic dimension over the torus] \label{prop:genericDimension}
  Let $\bm A=(A_1,\ldots, A_r)$ be a family of $r$ nonempty finite subsets of
  $M$. Let
  $f_1,\dots,f_r \in \C[M]$ be an extremal-generic
  system. Then one of the
  two following propositions holds true:
  \begin{itemize}
  \item The family $\mathbf A$ is \emph{essential}
    and
    $\dim(\C[M] / \ideal{f_1,\dots,f_r}) = n - r$;
  \item The family $\mathbf A$ is \emph{not essential} and
    $\ideal{f_1,\dots,f_r} = \C[M]$.
  \end{itemize}
\end{proposition}

\begin{proof}
  First, assume that $\bm{A}$ is essential.  By noticing that the
  mixed volume of the convex hulls of $A_1,\ldots, A_r$ completed with
  $n-r$ standard simplices is
  nonzero~\cite[Thm.~5.1.7]{schneider2014convex}, we deduce from the
  BKK theorem~\cite{bernshtein_number_1975} that the intersection of
  the variety associated to $f_1,\ldots, f_r$ with a generic affine
  space of dimension $r$ is a finite number of points in $(\C^*)^n$.
  We emphasize that the reason why the genericity condition needs only
  to be put on the coefficients of the system which correspond to
  vertices of the Newton polytopes follows from the works \cite[Thm.~B,
  (a)]{bernshtein_number_1975} and \cite{canny1991optimal}.  This
  implies that $\dim(\ideal{f_1,\dots,f_r}) = n - r$ since the
  intersection of a variety $V$ in $(\C^*)^n$ with a generic linear
  space of dimension $d$ must either be empty or have dimension
  $\dim(V)-n+d$.

  Conversely, assume that $\bm{A}$ is not essential. Then there is a subset
  $E\subseteq \conSet{1, r}$ such that
  $\dim(\AffineSpan_\R(\sum_{j \in E} A_j)) < \lvert E\rvert$ which implies that the
  rank of the lattice $M_E\subseteq \Z^n$ spanned
  by $\cup_{j\in E} A_j$ is smaller than $\lvert E\rvert$.
  Up to reordering, we can assume without loss of generality that $E =
  \conSet{1,k}$ for some $k\leq r$.
  Since $M_E$ is
  isomorphic to $\Z^{\rank(M_E)}$ as a $\Z$-module,
  Lemma~\ref{lem:technical_lemma_sturmfels} implies that for any values of the
  non-extremal coefficients, the set of systems
  in $\C[M_E]$ with solutions in $(\C^*)^{\rank(M_E)}$ and support
  $(A_1,\dots,A_k)$ is contained in a proper hypersurface of
  $\A^{\lvert V_1\rvert}\times\dots\times\A^{\lvert V_k\rvert}$, where $V_i$ is
  the set of extremal points in $A_i$. Said
  otherwise, for an extremal-generic system $f_1,\ldots, f_k\in \C[M_E]$ with support
  $(A_1,\dots,A_k)$, by Hilbert's Nullstellensatz there exist polynomials
  $\{g_j\}_{j\in E}$ in $\C[M_E]$ such that $\sum_{j\in E} f_j g_j = 1$.
  This relation also holds in $\C[M]$, which
  shows that $\ideal{f_1,\ldots, f_r} = \ideal{f_1,\ldots, f_k} = \C[M]$ for extremal-generic
  systems.
\end{proof}

\begin{remark*}
  A similar criterion as the one in
  Proposition~\ref{prop:genericDimension} appeared in
  \cite{yu2016most}, where the authors studied when these ideals are
  prime. Our result strengthens their criterion by providing the Krull
  dimension of the associated quotient ring.
\end{remark*}

We will now extend the notion of extremal-generic polynomials to the
Cox ring.

\begin{definition}\label{def:Coxringextremal}
  Let $\alpha\in \Cl(X)$ be a divisor class, and $D$ be a $T$-divisor with class
  $\alpha$. Then, by \Cref{prop:bij_pol_Sa}, the monomials in
  $S_\alpha$ are in bijection with $A_{D,\{0\}}$. Via this bijection,
  the \emph{extremal coefficients} of a
  polynomial with monomial support $\mathcal A\subset S_\alpha$ are
  those corresponding to the extremal points in $A_{D,\{0\}}$.
  We say that a property holds for an extremal-generic system of
  homogeneous polynomials $f_1,\dots,f_r$ with supports
  $\A_1,\dots,\A_r$ if for any fixed values of the non-extremal coefficients the
  set of values of the extremal coefficients satisfying the property is
  Zariski-dense.
\end{definition}

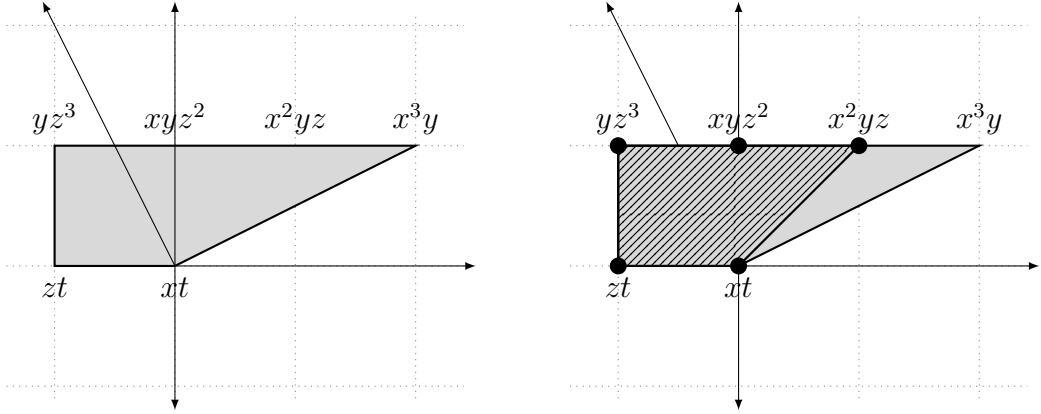
\begin{figure}
  \centering
  \begin{tikzpicture}[scale=1.6]
\draw[step=1cm,gray,thin, dotted] (-1.4,-1.1) grid (2.4,2.1);
  \filldraw[thick, fill= gray!30!white] (0,0) node[anchor=north] {$xt$} -- 
  (2, 1) node[anchor=south] {$x^3y$} -- 
  (1, 1) node[anchor=south] {$x^2yz$} -- 
  (0, 1) node[anchor=south] {$xyz^2$} -- 
  (-1, 1) node[anchor=south] {$yz^3$}-- 
  (-1, 0) node[anchor=north] {$zt$} -- (0, 0);
\draw[thin, -latex] (0,0) -- (2.5,0);
\draw[thin, -latex] (0,0) -- (0,2.2);
\draw[thin, -latex] (0,0) -- (0,-1.2);
\draw[thin, -latex] (0,0) -- (-1.1,2.2);
\end{tikzpicture}\quad\quad\quad
  \begin{tikzpicture}[scale=1.6]
\draw[step=1cm,gray,thin, dotted] (-1.4,-1.1) grid (2.4,2.1);
\draw[thin, -latex] (0,0) -- (2.5,0);
\draw[thin, -latex] (0,0) -- (0,2.2);
\draw[thin, -latex] (0,0) -- (0,-1.2);
\draw[thin, -latex] (0,0) -- (-1.1,2.2);
  \filldraw[thick, fill= gray!30!white] (0,0) node[anchor=north] {$xt$} -- 
  (2, 1) node[anchor=south] {$x^3y$} -- 
  (1, 1) node[anchor=south] {$x^2yz$} -- 
  (0, 1) node[anchor=south] {$xyz^2$} -- 
  (-1, 1) node[anchor=south] {$yz^3$}-- 
  (-1, 0) node[anchor=north] {$zt$} -- (0, 0);
  \filldraw[thick, pattern=north east lines] (0,0) -- 
  (1, 1)  -- 
  (0, 1)  -- 
  (-1, 1)  -- 
  (-1, 0)  -- (0, 0);
    \fill[fill=black] (0,0) circle (2pt) -- 
  (1, 1)  circle (2pt)-- 
  (0, 1)  circle (2pt)  -- 
  (-1, 1)  circle (2pt)-- 
  (-1, 0)  circle (2pt)-- (0, 0);
\end{tikzpicture}
  \caption{The two-dimensional fan associated to the Hirzebruch surface
  $\mathcal H_2$. On the left, the
  polytope $P_{D,\{0\}}$ associated to the divisor $D = D_1 + D_4$ where $D_1$
  is the $T$-divisor of the ray $\R_{\geq 0}(1,0)$ and $D_4$ is the $T$-divisor of the ray
  $\R_{\geq 0}(0, -1)$. The lattice points of $P_{D,\{0\}}$ are in bijection with the
  monomials in the Cox ring $\C[x,y,z,t]$ of degree $(1,1)$, where $\deg(x^a
  y^bz^ct^d)=(c+a-2b, d+b)$. On the right, the Newton polytope of the
  polynomial in \Cref{example:Coxring} inside $P_{D,\{0\}}$.\label{figure:hirzebruch}}
\end{figure}

\begin{example}\label{example:Coxring}
  The toric variety associated to the fan in \Cref{figure:hirzebruch} defines the Hirzebruch surface
  $\mathcal H_2$. Its Cox ring is $\C[x,y,z,t]$ with a
  $\Z^2$-grading given by $\deg(x^a\,y^b\,z^c\,t^d) = (c+a-2b, d+b)$ and
  irrelevant ideal $B=\langle xy, yz, zt, xt\rangle$. Consider the
  homogeneous polynomial $f = yz^3+2zt+3xt-xyz^2+4x^2yz$ of degree $(1,1)$ and
  support $\mathcal A =\{yz^3, zt, xt, xyz^2, x^2yz\}$, which is a proper subset
  of the six monomials
  of degree $(1,1)$ in $S$ which correspond to the lattice points in the polytope in
  \Cref{figure:hirzebruch}. Then the extremal coefficients
  of $f$ are those associated to the monomials $\{yz^3, zt, xt, x^2yz\}$, which are the vertices
  of the convex hull of the support.
\end{example}

Our next objective is to extend the notion of \emph{essential family}
to monomial sets $\{\mathcal A_1, \ldots, \mathcal A_r\}$ such that
$\mathcal A_i \subset S_{\alpha_i}$ for all $i$. For this, we
introduce the following notation.

\begin{notation} 
  Consider a monomial set $\mathcal A \subset S_\alpha$ and let
  $\sigma$ be a cone in $\Sigma$. We define the set of points
  $\mathcal A^\sigma$ as the exponents in $M$ of the monomial set
  $\{f^\sigma : f \in \mathcal A\} \subset \C[M\cap \sigma^\perp]$,
  where $f^\sigma$ is the image of $f$ via the map in
  \Cref{prop:map_compatible}.  We recall that $\mathcal A^\sigma$ is only
  defined up to translation by elements in $\sigma^\perp\cap M$, see
  the discussion after \Cref{prop:map_compatible}.
  We emphasize that $\mathcal A^\sigma$ might be empty, even
  though $\mathcal A$ is not.
\end{notation}

\begin{theorem}\label{thm:dimIntersectionToricTorusOrbit}
  For $j\in\conSet{1,r}$, fix a divisor class $\alpha_j \in \Cl(X)$ and a
  monomial subset $\mathcal A_j\subset S_{\alpha_j}$. Let $Y$ be a closed
      subscheme of $X$ defined by a homogeneous ideal
      $\ideal{f_1,\dots,f_r}\subset S$, where each $f_i \in S_{\alpha_i}$ is
      extremal-generic of degree $\alpha_i$ with support $\mathcal A_i$. For
      a cone $\sigma$ in $\Sigma$, set $E_\sigma := \{i\in \conSet{1, r} \mid
  \mathcal A_i^\sigma \ne\emptyset\}$. Then
      \begin{itemize}
      \item if the family $\{\mathcal A_{i}^\sigma : i\in E_\sigma \}$
        is not essential, then
        $Y\cap O(\sigma)$ is empty;
	\item otherwise, the dimension of $Y \cap O(\sigma)$ is
          $n - \dim(\sigma) - \lvert E_\sigma \rvert$.
  \end{itemize}
\end{theorem}
    
\begin{proof}
    By Corollary~\ref{coro:scheme_intersection}, the ideal
  $\ideal{ f^\sigma_i : i \in E_\sigma} \subset \C[M \cap
  \sigma^\perp]$ defines $Y \cap O(\sigma)$. Moreover, the polynomials
  $f^\sigma_i\in \C[\sigma^\perp\cap M]$ have monomials with exponents
  in $\mathcal A_i^\sigma$ and extremal-generic coefficients by
  Proposition~\ref{prop:map_compatible}.
  As the dimension of $M \cap \sigma^\perp$ is $n - \dim(\sigma)$, by
  \Cref{prop:genericDimension} either $Y \cap O(\sigma)$ is empty or
  $\{\mathcal A_{i}^\sigma : i \in E_\sigma\}$ is essential.
  In the latter case, as each $\mathcal A_{i}^\sigma$ belongs to
  $M \cap \sigma^\perp$, which has dimension
  $n - \dim(\sigma)$, we have that 
  $n - \dim(\sigma) \geq \lvert E_\sigma\rvert$.
  Hence, if $\{\mathcal A_{i}^\sigma : i \in E_\sigma\}$ is essential,
  $Y \cap O(\sigma)$ has dimension
  $n - \dim(\sigma) - \lvert E_\sigma\rvert \geq 0$.
    \end{proof}

    Using~\Cref{thm:dimIntersectionToricTorusOrbit}, we can study the dimension of a closed
    subscheme $Y$ generated by extremal-generic homogeneous
    polynomials of degrees $\alpha_1,\ldots,\alpha_r$ by looking at
    which families $\{\mathcal A_{i}^\sigma\}_{i \in E_\sigma}$ are essential
    as $\sigma$ ranges through all the cones of $\Sigma$. The
    following theorem is the technical core of the paper.

\begin{theorem}\label{thm:dimensionGenericToricSystem}
  For $j\in\conSet{1,r}$, fix a divisor class $\alpha_j \in \Cl(X)$ and a
  monomial subset $\mathcal A_j\subset S_{\alpha_j}$. Let $f_j\in S_{\alpha_j}, E_\sigma, Y$ be as in
  Theorem~\ref{thm:dimIntersectionToricTorusOrbit}.
  For each cone $\sigma$ in $\Sigma$, let
   $E_\sigma = \{i\in \conSet{1, r} \mid
  \mathcal A_i^\sigma \ne\emptyset\}$.
  By abuse of notation, we say that $E_\sigma$ is essential if
  $\{\mathcal A_i^\sigma : i \in E_\sigma\}$ is essential.
    Let $Y$ be a closed subscheme of $X$ defined by a homogeneous ideal
    $\ideal{f_1,\dots,f_r}\subset S$, where each
    $f_i \in S_{\alpha_i}$ is extremal-generic of degree $\alpha_i$
    with support $\mathcal A_i$.
  Then the dimension of $Y$ is
  $$
  \dim(Y)
  =
  \max_{\substack{\sigma \in \Sigma \\ E_\sigma \text{ is essential}}} \left(n
  - \dim(\sigma) - \lvert E_\sigma\rvert \right).
  $$
  If none of the sets $E_\sigma$ is essential, then $Y$ is empty.
\end{theorem}
    \begin{proof}
      Following the orbit-cone correspondence, we write
      $X = \bigcup_{\sigma \in \Sigma} O(\sigma)$ as a finite union of
      torus orbits. Since $Y$ is a subscheme of $X$, it can be written as a
      finite union of locally closed subschemes of $X$:
      $$Y =\bigcup_{\sigma\in\Sigma} \left(Y \cap O(\sigma)\right).$$
      Consequently,
      $$\dim(Y) = \dim \left(Y \cap \left(\bigcup_{\sigma \in \Sigma}
          O(\sigma) \right)\right) = \max_{\sigma \in \Sigma} \left(\dim
      \left(Y \cap O(\sigma)\right)\right).
      $$
      By \Cref{thm:dimIntersectionToricTorusOrbit},
      $\dim(Y \cap O(\sigma)) = n - \dim(\sigma) - \lvert
      E_\sigma\rvert \geq 0$ if $\{\mathcal A_{i}^\sigma : i \in E_{\sigma}\}$
      is essential,
      otherwise $Y \cap O(\sigma)$ is empty.
    \end{proof}
    
\section{Subsystems and complete intersections}\label{sec:completeinter}

In the classical polynomial algebra $\C[x_0,\dots,x_n]$, the
codimension of the zero set of a generic homogeneous system
$f_1,\dots,f_r$ is $r$. Such a variety $Y$ is called a \emph{complete
  intersection}. Complete intersections have been greatly studied
from a computational viewpoint since this structural property provides
tools to speed-up computations and to estimate the complexity of
Gr\"obner basis algorithms \cite{bardet2015complexity}.

Motivated by the classical projective case, in this section we study
subschemes $Y$ of $X$ associated to ideals generated by $\codim(Y)$ elements in
the Cox ring of $X$.

\begin{definition}[Complete intersection]
  Let $r\leq n$. We say that a system
  $(f_1,\dots,f_r) \in S_{\alpha_1} \times \dots \times S_{\alpha_r}$
  defines a \emph{complete intersection} if the subscheme $Y$ of $X$
  associated to the ideal generated by $f_1,\dots,f_r$ is nonempty and has dimension
  $n - r$.
\end{definition}

We would like to emphasize that some care is required with this
definition; indeed, $X$ need not be smooth, and this implies that
classical results on complete intersections might not apply to our
context. However, our complete intersections are, in particular,
set-theoretic complete intersections.
In order to avoid part of this difficulty, in this section we restrict the
analysis to effective divisor classes
$\alpha_1,\dots,\alpha_r$ which are $\Q$-Cartier, i.e. there exist multiples of these classes which are Cartier.
Complete intersections defined by Cartier divisors had been extensively
studied; for example, we can think about the BKK bound
\cite{bernshtein_number_1975} as the degree of a zero-dimensional
complete intersection \cite[Sect. 5.4]{fulton1993introduction}.

We start this section by characterizing the extremal-generic systems
leading to complete intersections via 
\Cref{thm:dimensionGenericToricSystem}.

\begin{corollary}\label{coro:complete_inter_criterion}
  Let $r \leq n$. For $j\in\conSet{1, r}$, fix an effective
  $\Q$-Cartier divisor class $\alpha_j$, and a monomial subset
  $\mathcal A_i\subset S_{\alpha_i}$. Let $Y$ be a closed subscheme of
  $X$ defined by a homogeneous ideal $\ideal{f_1,\dots,f_r}\subset S$,
  where each $f_i \in S_{\alpha_i}$ is extremal-generic of degree
  $\alpha_i$ with support $\mathcal A_i$. For a cone $\sigma$ in
  $\Sigma$, set
  $E_\sigma = \{i\in \conSet{1, r} \mid \mathcal A_i^\sigma
  \ne\emptyset\}$. Then, $Y$ is a complete intersection in $X$ if and
  only if the following conditions are satisfied:
  \begin{itemize}
    \item for all $\sigma\in \Sigma$, either $E_\sigma$ is not essential or
      $\dim(\sigma) + \lvert E_\sigma \rvert \geq r$;
   \item there exists $\sigma\in\Sigma$ such that
  $E_{\sigma}$ is essential.
  \end{itemize}
\end{corollary}

\begin{proof}
  If $Y$ is a complete intersection, then $\dim(Y)=n-r$.
  By Theorem~\ref{thm:dimensionGenericToricSystem}, for all
  $\sigma\in \Sigma$, either $E_\sigma$ is not essential or
  $\dim(\sigma) + \lvert E_\sigma \rvert \geq r$. Since $Y$ is
  not empty, there exists $\sigma\in\Sigma$ such that $E_{\sigma}$ is
  essential.
  
  To prove the converse implication we observe that the two conditions
  together with Theorem~\ref{thm:dimensionGenericToricSystem} imply
  that $Y$ is not empty and $\dim(Y) \leq n-r$.
  In what follows we show that, as the divisor classes as
  $\Q$-Cartier, then either $Y$ is empty or $\dim(Y) \geq
  n-r$. Putting all together, we conclude that $\dim(Y) = n-r$.
  
  First, let us assume that the divisor classes $\alpha_1,\dots,\alpha_r$
  are Cartier. Then at each affine chart
  $U_{\sigma}$, the subscheme
  $Y \cap U_{\sigma}$ is defined by at most $r$
  nonzero polynomials; see \Cref{lem:compatible}.
  By Krull's height
  theorem~\cite[Thm.~10.2]{eisenbud_commutative_2004},
  $\dim(Y \cap U_{\sigma}) \geq n-r$ for each maximal $U_\sigma$ such
  that $Y \cap U_{\sigma}$ is not empty. Therefore, if the divisor classes
  are Cartier, $\dim(Y) \geq n-r$ or $Y$ is empty.
  If the divisor classes are not Cartier, but $\Q$-Cartier, let
  $k \in \N$ be such that $k \, \alpha_1,\dots, k \, \alpha_r$ are
  Cartier. Let $X'$ be the toric variety defined by the same fan as
  $X$, but with respect to the lattice $M ' := \frac{1}{k} M$, instead
  of the original lattice $M$.
  By \cite[Prop.~3.3.7, Ex.~5.0.13]{cox2011toric}, there is a toric
  morphism $\phi: X' \rightarrow X$ which presents $X$ as a geometric
  quotient of $X'$ by a finite group. In particular, the dimensions of
  $Y$ and $\phi^{-1}(Y)$ are the same and we can regard the original
  divisors $\alpha_1,\dots,\alpha_r$ in $X$ as Cartier divisors on
  $X'$. Moreover, because $M' \subseteq M$, we have that each monomial
  set $\mathcal A_i \subset S_{\alpha_i}$ contains monomials of degree
  $k\alpha_i$ in the Cox ring of $X'$ and the notions of essentiality
  and extreme-genericity
  do not change.
  Hence, we can assume without loss of generality that
  $\alpha_1,\dots,\alpha_r$ are Cartier.
\end{proof}

In contrast with what happens in the classical projective case, being
a complete intersection does not imply that every subsystem is a
complete intersection. 

\begin{example} \label{ex:subsystemProblematic}
  Consider the toric variety $\PP^1 \times \PP^2$ and its Cox ring
  $\C[x_0,x_1] \otimes \C[y_0,y_1,y_2]$ equipped with its natural $\mathbb Z^2$
  grading. Set $f_1 = x_0$, $f_2 = x_1\, y_0$, $f_3 = x_1\, y_1$.
  Then, the sequence $(f_1, f_2, f_3)$ defines a complete intersection,
  but $(f_2, f_3)$ does not.
\end{example}

Under the $\Q$-Cartier assumption, the following statement characterizes in which cases every subsystem of a
extremal-generic complete intersection also defines a complete
intersection.

\begin{lemma}
  \label{lem:genericallyGoodSubsys}
For $j\in\conSet{1,r}$, fix an effective $\Q$-Cartier divisor class $\alpha_j \in \Cl(X)$ and a
  monomial subset $\mathcal A_j\subset S_{\alpha_j}$.  
Let $Y$ be a closed
  subscheme of $X$ defined by a homogeneous ideal
  $\ideal{f_1,\dots,f_r}\subset S$, where each $f_i \in S_{\alpha_i}$
  is extremal-generic of degree $\alpha_i$ with support
  $\mathcal A_i$.
  Assume that $Y$ is not empty.
  For each subset $I \subseteq \conSet{1, r}$, let $Y_I$ be the closed
  subscheme
  associated to $\ideal{f_{i}}_{i\in I}$.
  For each cone $\sigma$ in $\Sigma$, let
  $E_\sigma = \{ i\in\conSet{1, r}\mid \mathcal A_i^\sigma
  \neq \emptyset\}$.
  Then, for every $I \subseteq \conSet{1, r}$, $\dim(Y_I) = n - \lvert I\rvert$ if
  and only if,
  for each cone $\sigma \in \Sigma$,
  $r - \lvert E_\sigma\rvert \leq \dim(\sigma)$.
\end{lemma}
\begin{proof}
  First, we show that $\dim(Y_I) = n-\lvert I\rvert$ holds true if and only if
  for all $\sigma\in\Sigma$ such that $E_\sigma \cap I$ is essential,
  the inequality
  $\lvert I\rvert-\lvert E_\sigma\cap I\rvert\leq \dim(\sigma)$ holds
  true.
  As $Y_I$ contains $Y$, it is nonempty; therefore, by
  Theorem~\ref{thm:dimensionGenericToricSystem}, there is at least one
  essential family $E_\sigma \cap I$.
  By the same theorem, $\dim(Y_I)$ is equal to the maximum of
  $n - \dim(\sigma) - \lvert E_\sigma \cap I \rvert$ among the sets
  $E_\sigma \cap I$ which are essential.
  Hence, if there is an essential $E_\sigma \cap I$ such that
  $\lvert I\rvert-\lvert E_\sigma\cap I\rvert > \dim(\sigma)$, then
  $\dim(Y_I) > n-\lvert I\rvert$.
  Otherwise, $\dim(Y_I) \geq n-\lvert I\rvert$.
  Moreover, in this case $Y_I$ is defined by $\lvert I\rvert$
  polynomials in $S$ with $\Q$-Cartier degrees, so by Krull's height
  theorem (see the proof of \Cref{coro:complete_inter_criterion}), we
  have that $\dim(Y_I) \geq n-\lvert I\rvert$. Therefore,
  $\dim(Y_I) = n-\lvert I\rvert$.
    
  To prove the if direction of the lemma, notice that for any subset
  $I\subseteq\conSet{1, r}$ and any $\sigma \in \Sigma$, we have that
  $E_\sigma \subseteq \conSet{1, r}$, and so
  $\lvert I\rvert-\lvert E_\sigma\cap I\rvert\leq r-\lvert
  E_\sigma\rvert$ (when $I = \conSet{1, r}$, this is an equality).
  Hence, if $r - \lvert E_\sigma\rvert \leq \dim(\sigma)$ for any
  cone $\sigma$ in $\Sigma$, then the inequality
  $\lvert I\rvert-\lvert E_\sigma\cap I\rvert\leq \dim(\sigma)$ holds
  for all $I\subseteq \conSet{1, r}$ and for all cone $\sigma$ in $\Sigma$. Therefore, $\dim(Y_I) = n-\lvert I\rvert$.

  To prove the only if direction, assume reciprocally that there is
  $\sigma \in \Sigma$ such that
  $r - \lvert E_\sigma\rvert > \dim(\sigma)$. Let $I := \conSet{1, r} \setminus E_\sigma$.
  Then, $\lvert I \rvert = r - \lvert E_\sigma\rvert > \dim(\sigma)$ and
  $E_\sigma \cap I = \emptyset$ is essential. Hence,
  by Theorem~\ref{thm:dimensionGenericToricSystem},
  $\dim(Y_I) \geq n - \dim(\sigma) - \lvert E_\sigma \cap I \rvert > n-\lvert I\rvert$.
\end{proof}

\begin{example}[Cont. Example~\ref{ex:subsystemProblematic}]
  Consider the $1$-dimensional cone $\sigma$ associated to the torus orbit
  $O(\sigma) = \{((1:0),(a : b)) \in \PP^1 \times \PP^2 \mid a\ne 0, b\ne 0\}$.
  Using the notation from \Cref{lem:genericallyGoodSubsys}, $\lvert
  E_\sigma\rvert = 1$ since $f_2$
  and $f_3$ vanish on $O(\sigma)$. Hence, $3 - 1 > \dim(\sigma) = 1$, and
  therefore \Cref{lem:genericallyGoodSubsys} states that there must exist a
  subsystem of $(f_1, f_2, f_3)$ which does not define a complete intersection.
\end{example}

An important special case of \Cref{lem:genericallyGoodSubsys} arises when the
divisors are
Cartier and nef. In that case, under the extremal-genericity assumption, every
subsystem has the expected
codimension. In our setting, nef Cartier divisor classes correspond to
isomorphism classes of basepoint-free line bundles on our toric variety,
see~\cite[Thm.~6.3.12]{cox2011toric}.

\begin{corollary}
  Let $r\leq n$ and $\alpha_1,\dots,\alpha_r \in \Cl(X)$ be nef
  Cartier divisor classes. For each $i\in\conSet{1, r}$, let $\A_i$ be the set of
  monomials in $S_{\alpha_i}$.
  Let $Y$ and $Y_I$ (for $I\subseteq \conSet{1,r}$) be as in Lemma~\ref{lem:genericallyGoodSubsys}, and
  assume that $Y$ is nonempty. Then
  $\dim(Y_I) = n - \lvert I\rvert$ for every
  $I \subseteq \conSet{1, r}$.
\end{corollary}

\begin{proof}
  If $\alpha_i$ is a nef Cartier divisor class, then $\alpha_i$ is
  basepoint-free \cite[Thm. 6.3.12]{cox2011toric}, so for each
  $\sigma \in \Sigma$ there is a monomial in $m \in S_{\alpha_i}$ such
  that $m^\sigma \neq 0$ \cite[Prop.
  6.1.1]{cox2011toric}.
  Hence for every cone $\sigma$ in $\Sigma$, $E_\sigma = \conSet{1, r}$.
  As $Y$ is nonempty ,
   the corollary follows from
  Lemma~\ref{lem:genericallyGoodSubsys}.
\end{proof}

One could ask if the conditions of the previous lemma allow us to
derive any information about subsystems of arbitrary (non
extremal-generic) complete intersections. However, this is not true,
as can be checked by applying a (structured) generic linear change of
coordinates to the system from Example~\ref{ex:subsystemProblematic}.
However, when the degrees $\alpha_1,\dots,\alpha_r$ are associated to
$\Q$-ample divisors, we can drop the genericity assumptions.
By slight abuse of notation, we say that a divisor class $\alpha$ is
\emph{$\Q$-ample} if it is $\mathbb Q$-Cartier and the line bundle
associated to a Cartier multiple of $\alpha$ is ample.

\begin{lemma} \label{lem:intQample}
  Let $Y$ be a closed subscheme of $X$ of dimension $\geq 1$ and $\alpha\in\Cl(X)$ be a $\Q$-ample
  divisor class. Consider $f \in S_\alpha$ and let $Z$ be the closed
  subscheme of $X$ associated to $\langle f \rangle$.  Then the
  scheme-theoretic intersection $Y \cap Z$ is nonempty and its dimension satisfies the following
  inequalities:
  \[
    \dim(Y) \geq \dim(Y \cap Z) \geq \dim(Y) - 1.
  \]
\end{lemma}
\begin{proof}
  First we observe that $Z$ is a
  $\Q$-ample divisor on $X$ with class $\alpha \in
  \Cl(X)$. As we did in the proof of
  \Cref{coro:complete_inter_criterion}, up to replacing $M$ by
  $\frac{1}{k} M$ for some $k\in\Z_{> 0}$, we can assume without loss
  of generality that the sheaf $\mathscr O(Z)$ (see
  \cite[Prop.~4.0.27]{cox2011toric}) is a very ample line bundle and
  $Z$ is a very ample divisor.
 
  Let $\phi : X \xhookrightarrow{} \mathbb P^N $ be the closed
  immersion of $X$ into an $N$-dimensional projective space defined by
  the global sections of $\mathscr O(Z)$.
  By \cite[Thm.~II.7.1]{hartshorne},
  there is a hyperplane $H$ in $\mathbb P^N$ such that
  $\phi(Z) = \phi(X) \cap H$.
  Hence, $\phi(Y \cap Z) = \phi(Y) \cap \phi(Z) = \phi(Y)
  \cap H$. 
  By the projective dimension theorem
  \cite[Thm.~I.7.2]{hartshorne},
  $$\dim(\phi(Y)) \geq \dim(\phi(Y \cap Z)) = \dim(\phi(Y) \cap H) \geq \dim(\phi(Y)) - 1.$$
  Finally, our lemma follows from the fact that closed immersions
  preserve dimensions and 
  so $\dim(\phi(Y))=\dim(Y)$ and $\dim(\phi(Y \cap Z))=\dim(Y\cap Z)$.
\end{proof}

\begin{theorem}\label{thm:Qample_subsets}
  Let $r \leq n$ be an integer, $\alpha_1,\ldots, \alpha_r\in\Cl(X)$ be $\Q$-ample
  divisor classes, and
  $(f_1,\ldots, f_r)\in S_{\alpha_1}\times\dots\times S_{\alpha_r}$ be
  a sequence of arbitrary polynomials. Then the closed subscheme of $X$
  defined by $(f_1,\ldots, f_r)$ is nonempty. Moreover, if
  $(f_1,\ldots, f_r)$ defines a complete intersection in $X$, then any
  subsystem of $(f_1,\ldots, f_r)$ also defines a complete
  intersection.
\end{theorem}

\begin{proof}
  Let $Y_i$ be the closed subscheme of $X$ associated to $(f_1,\dots,f_i)$.
  By Lemma~\ref{lem:intQample},
  $\dim(Y_{i-1}) \geq \dim(Y_i) \geq \dim(Y_{i-1}) - 1$ for all
  $i\in\conSet{1,r}$. By induction, we have that $n \geq \dim(Y_i) \geq n - i$ for all
  $i\in\conSet{1,r}$. In particular, if $\dim(Y_i) = n - i$, then
  $n - i + 1 = \dim(Y_i) + 1 \geq \dim(Y_{i-1}) \geq n - i - 1$, so
  $\dim(Y_{i-1}) = n - i - 1$.
  Proceeding inductively, this shows that, if $\dim(Y_r) = n - r$,
  that is, if $Y_r$ is a complete intersection, then
  $\dim(Y_i) = n - i$, for every $i$. The rest of the proof follows
  from the fact that the relative order of the polynomials
  $(f_1,\dots,f_r)$ does not change $Y_r$ (nor its dimension), so the
  same argument holds up to reordering the polynomials.
\end{proof}

\begin{remark*}
  We emphasize that the assumptions of
  \Cref{thm:Qample_subsets} implicitly
  require $X$ to be projective. Indeed, a
  toric variety $X$ defined by a complete fan is projective if and only if there
  exists a very ample line bundle on $X$, which is equivalent to the existence of a
  $\Q$-ample divisor class.
\end{remark*}

\section{Polytopal algebras}\label{sec:polytopalAlgebra}

Polytopal algebras provide us with 'nice' algebras to compute with
sparse polynomials: they are $\N$-graded and Cohen-Macaulay
\cite[Thm.~1]{hochster1972rings}. As we shall see in the proof of \Cref{thm:regSeqInPolytopalAlg}, in this
case any homogeneous
system defining a complete intersection is regular. This property can be
used for instance for Gr\"obner basis computations
\cite{faugere_complexity_2016,faugere2014sparse,bender2019algorithms}.
In this section, we characterize the degrees at which extremal-generic systems are regular.

Let $P \subset M_\R$ be a rational full-dimensional polytope, i.e. a polytope
whose vertices lie in $M\otimes \Q$. We let $\Sigma$ denote its normal
fan, and $\sigma\subset M_\R\times \R$ be the pointed cone $\R_{\geq 0}(P\times
\{1\})\subset \R^{n+1}$. We
consider the $\Z_{\geq 0}$-graded polytopal algebra
$$
\C[P_M] = \bigoplus_{i \geq 0} \bigoplus_{m \in ((i\cdot P)\times\{i\}) \cap
(M\times \Z)}
\C \cdot \chi^{m}.
$$
Said otherwise, monomials in $\C[P_M]_i$ can be identified with
lattice points in $i\cdot P$.  An important fact is that
$\Proj(\C[P_M])$ is isomorphic to the toric variety associated to the
normal fan of $P$: this is stated in \cite[Thm.~7.1.13]{cox2011toric} for
lattice polytopes, and it can be extended to rational polytopes by
using the Veronese embedding, see \cite[Sec.~3]{reid2002graded}.
In fact, up to
refining the lattice $M$, the polytopal algebra $\C[P_M]$ can be thought of as
the subalgebra of the Cox ring which is generated by homogeneous elements whose
degrees correspond to $\Q$-ample divisor classes with associated polytope equal to $P$ (up
to translation).

\subsection{Extremal-generic systems in $\C[P_M]$}

We start this section by proving that we can assume, without loss of
generality, that $P$ is a normal lattice polytope.  For this, let $M'$
be a larger lattice such that $[M':M]<\infty$ and such that $P$ is a
normal lattice polytope over $M'$; such an $M'$ exists
by~\cite[Thm.~2.2.12]{cox2011toric}. Then there is a natural group
action of $G = \Hom(M'/M, \C^*)$ on the torus $T = \Hom(M', \C^*)$
which extends to toric varieties defined over $M'$. The inclusion of
$M$ is $M'$ gives rise to surjective morphisms of toric varieties
defined over $M'$ onto their counterparts defined over $M$, which
present them as geometric quotient by $G$, see
e.g.~\cite[Prop.~1.3.18, Prop.~3.3.7,
Example~5.0.13]{cox2011toric}. This is the general framework of this
section: we consider a refinement of the lattice and we show that this
does not change the dimension of the associated subscheme.

We observe that $G$ acts on $\C[P_{M'}]$ via the morphism
$G\rightarrow \Aut(\C[P_{M'}])$ which sends $(g, \chi^{(m', i)})$ to
$g(m')\chi^{(m',i)}\in\C[P_{M'}]$. The ring of invariants $\C[P_{M'}]^G$ equals
the image of $\C[P_M]$ via the inclusion $\C[P_M]\hookrightarrow\C[P_{M'}]$,
and hence it is generated by monomials.

\begin{proposition}\label{prop:changelattice}
  Let $\pi:\Proj(\C[P_{M'}])\rightarrow \Proj(\C[P_{M}])$ be the surjective
  morphism dual to the inclusion morphism $\C[P_M]\hookrightarrow \C[P_{M'}]$.
  For any closed subscheme $Y\subset \Proj(\C[P_{M}])$, the
  scheme-theoretic inverse image $\pi^{-1}(Y)$ has the same dimension as $Y$.
\end{proposition}

\begin{proof}
 If $I\subset \C[P_M]$ is a homogeneous
  ideal defining a closed subscheme $Y$ of $\Proj(\C[P_M])$, then
  $I\otimes_{\C[P_M]}\C[P_{M'}]$ defines the scheme-theoretic inverse image
  $\pi^{-1}(Y)$~\cite[Def. after Example
7.12.1]{hartshorne}. The dimension of $Y$ is the largest integer $d$ such that
  there exists a chain of inclusions of saturated prime homogeneous ideals $I\subset\mathfrak
  p_0\subsetneq\cdots\subsetneq \mathfrak p_d\subset \C[P_M]$.

  Let $\mathfrak p\subset \C[P_M]$ be a prime ideal, and $\sqrt{\mathfrak p\otimes
  \C[P_M']} = \cap_{i\in I} \mathfrak q_i$ be a minimal prime decomposition of
  the radical. Since $\mathfrak p\otimes
  \C[P_M']$ is invariant under the $G$-action, the minimal primes are globally
  invariant under the $G$-action. Using an arbitrary numbering of the orbits,
  we write
  the orbit decomposition $\sqrt{\mathfrak p\otimes
  \C[P_M']} = \cap_{1\leq i\leq \ell} \mathfrak o_i$, where $\mathfrak o_i\subset \C[P_{M'}]^G$ is the
  intersection of the prime ideals in the $i$-th orbit. If $\ell\geq 2$, then
  pick $x\in \mathfrak o_1$, $y\in \mathfrak o_2\cap \cdots\cap \mathfrak
  o_\ell$, such that $y\notin\mathfrak o_1$ and $x\notin \mathfrak o_2\cap \cdots\cap \mathfrak
  o_\ell$. Then $\prod_{g\in G} g\cdot (xy)^k\in \mathfrak p\otimes\C[P_{M'}]^G$
  for some $k>0$ but $\prod_{g\in G} g\cdot x^k\in \C[P_{M'}]^G,\prod_{g\in G} g\cdot y^k\in \C[P_{M'}]^G$,
  $\prod_{g\in G} g\cdot x^k,\prod_{g\in G} g\cdot y^k\notin\mathfrak
  p\otimes\C[P_{M'}]$, which contradicts the primality of $\mathfrak p$.
  Therefore $\ell=1$, which means that the set of minimal prime ideals in $\C[P_{M'}]$
  containing $\mathfrak p\otimes\C[P_{M'}]$ is a $G$-orbit.

  To conclude the proof, we must check that for any ideal
  $I\subset\C[P_M]$, $\dim(I\otimes \C[P_M'])=\dim(I)$. To
  this end, we consider an inclusion $\mathfrak p_1\subsetneq\mathfrak p_2$ of
  prime ideals in $\C[P_M]$. Then $\sqrt{\mathfrak p_1\otimes\C[P_{M'}]}\subsetneq
  \sqrt{\mathfrak p_2\otimes\C[P_{M'}]}$. Let $\mathfrak q_1$ be a minimal prime of $\mathfrak p_1\otimes
  \C[P_{M'}]$. Then it must contain a minimal prime $\mathfrak q_2$ of $\mathfrak p_2\otimes
  \C[P_{M'}]$. By the argument above, these minimal primes are
  $G$-orbits, which must be distinct; therefore $\mathfrak q_1\neq \mathfrak
  q_2$. This shows that for any chain $I\subset \mathfrak p_1\subsetneq \mathfrak
  p_2\subset \C[P_M]$, there exists a chain of prime ideals $(I\otimes\C[P_{M'}])\subset \mathfrak q_1\subsetneq \mathfrak
  q_2\subset \C[P_{M'}]$. This argument can be generalized to larger chains of
  prime ideals, which implies that the
  dimension of the closed subscheme of $\Proj(\C[P_M])$ defined by $I$ equals
  the dimension of its scheme-theoretic inverse image by $\pi$.
\end{proof}

Observe that the aforementioned map $\C[P_M]\hookrightarrow\C[P_{M'}]$
sends monomials to monomials. This shows that extremal-generic
homogeneous systems in $\C[P_M]$ can be thought of as extremal-generic
systems in $\C[P_{M'}]$.  Hence, we obtain the following statement.

\begin{theorem}\label{thm:polytopalcomplinter}
  Let $d_1,\ldots, d_r\in\N$ be degrees, and $A_1,\ldots, A_r$ be
  monomial subsets in $\C[P_M]_{d_1},\ldots, \C[P_M]_{d_r}$ i.e.
  $A_i\subset (d_i\cdot P)\cap M$.
  Let $f_1,\ldots, f_r$ be a homogeneous extremal-generic system of respective degrees
  $d_1,\ldots, d_r$ in $\C[P_M]$ with support $(A_1,\ldots, A_r)$.
  For any $F\in\Faces(P)\cup \{P\}$, set $E_F = \{i\in\conSet{1, r} \mid A_i\cap
  (d_i\cdot F)\ne \emptyset\}$. We say
  that $E_F$ is essential if the set $\{(A_i\cap (d_i\cdot F)) \mid i \in E_F\}$ is essential. Then the
  dimension of the closed subscheme of $\Proj(\C[P_M])$ defined by the ideal
  $\ideal{f_1,\ldots, f_r}$ is 
  \[\max_{\substack{F\in\Faces(P)\cup\{P\}\\\text{s.t. }E_F\text{ is
          essential}}}(\dim(F)-\card{E_F}).\]
If none of the subsets $E_F$ is essential, then the subscheme is empty.
\end{theorem}

\begin{proof}
  By \Cref{prop:changelattice}, the dimension of the closed subscheme defined
  by $\langle f_1,\ldots, f_r\rangle$ is the same as the dimension of the
  closed subscheme of $\Proj(P_{M'})$ defined by $\langle f'_1, \ldots,
  f'_r\rangle$, where $f'_i$ is the image of $f_i$ by the canonical map
  $\C[P_M]\hookrightarrow\C[P_{M'}]$.
  Then we notice that this map sends
  monomials to monomials, and therefore $f'_1,\ldots, f'_r$ is an
  extremal-generic system. Next, let $\Sigma$ be the normal fan of $P$, and let
  $X_{\Sigma, M'}$ be the associated toric variety. We notice that since $P_{M'}$
  has vertices in $M'$, it equals $P_{D, \{0\}}$ for some nonzero Cartier divisor
  $D\in\Div(X_{\Sigma, M'})$~\cite[Eq.~4.2.7, Prop.~4.2.10]{cox2011toric}. Then there is a canonical
  embedding $\C[P_{M'}]\hookrightarrow S$ which sends $\C[P_M']_d$ bijectively
  on $S_{d\,\alpha}$, where $S$ is the Cox ring of
  the toric variety $X_{\Sigma, M'}$ and $\alpha$ is the class of $D$
  \cite[Thm.~5.4.8.(c).(3)]{cox2011toric}.
  This map provides us with an isomorphism $\Proj(\C[P_{M'}])\simeq X_{\Sigma, M'}$ of
  schemes over $\C$ \cite[Prop.~3.1.6]{cox2011toric}. Under this isomorphism, we apply \Cref{thm:dimensionGenericToricSystem}
  to the system $f'_1,\ldots, f'_r$ in $S$, which yields the desired
  result.
\end{proof}

\subsection{Regular sequences in $\C[P_M]$}

Now that we characterized the expected codimension of a system defined
by extremal-generic polynomials, we will identify the degrees at which regular
sequences occur in polytopal algebras; here the main extra ingredient is the fact that
polytopal algebras are Cohen-Macaulay by Hochster's theorem~\cite[Thm.~1]{hochster1972rings}.

  \begin{theorem} \label{thm:regSeqInPolytopalAlg}
    Let $r \leq n$, $d_1,\ldots, d_r\in\N, A_1,\ldots, A_r, E_F$ be as in
    \Cref{thm:polytopalcomplinter}, and $f_1, \ldots, f_r$ be an
    extremal-generic system with support $A_1,\ldots, A_r$. Let $Y$ be
    the closed subscheme of $\Proj(\C[P_M])$ defined by the ideal
    $\langle f_1,\ldots, f_r\rangle$. Then:
  \begin{enumerate}
  \item The subscheme $Y$ is not empty.
    \item The sequence $f_1,\ldots, f_r$ is regular in $\C[P_M]$ if
      and only if for any
      $F\in\Faces(P)\cup\{P\}$,
      $\card{E_F}\geq
      \dim(F)+r-n$.
  \end{enumerate}
  \end{theorem}

Specializing the theorem by using the full monomial supports for $A_1,\ldots,
A_r$ leads to the following statement:

  \begin{corollary}
    Let $r \leq n$, $d_1,\ldots, d_r\in\N$ be positive integers, and $f_1, \ldots, f_r$ be
    a system of generic homogeneous polynomials in $\C[P_M]$ of respective
    degrees $d_1,\ldots, d_r$. The sequence
    $f_1,\ldots, f_r$ is regular if and only if for every
    $F\in\Faces(P)\cup\{P\}$,
    $$\card{\{i\in\conSet{1, r} \mid (d_i\cdot F) \cap M \ne
      \emptyset\}}\geq \dim(F)+r-n.$$
  \end{corollary}

  \begin{proof}
    As the system is generic, it is in particular an extremal-generic
    system with supports $A_1,\ldots, A_r$, where
    $A_i = d_i \, P \cap M$.
  \end{proof}

  We postpone the proof of \Cref{thm:regSeqInPolytopalAlg} to the end of this section, after we
  establish two useful combinatorial lemmas.
  The following statement is a direct consequence of
  \cite[Cor.~3.7]{bihan2019criteria}:

  \begin{lemma}\label{lemma:polytope_essential}
  Let $P\in M_\R$ be a full dimensional polytope, and $A_1,\ldots, A_r$ be finite subsets of $P$,
  for $r\leq n$. Assume that for every $F\in\Faces(P)$,
  $\card{E_F}>
  \dim(F)+r-n$, where $E_F := \{i\in\conSet{1, r} \mid A_i\cap F\ne \emptyset\}$. Then the
  family $A_1,\ldots, A_r$ is essential.
\end{lemma}

\begin{proof}
  Consider the family 
  $\{\Conv(A_1),\ldots, \Conv(A_r), P, \ldots, P\}$ where $P$ occurs
  $\dim(M)-r = n - r$ times. By~\cite[Cor.~3.7]{bihan2019criteria}, the mixed
  volume of this family equals the normalized volume of $P$, which
  implies that the family is essential
  \cite[Thm.~5.1.8]{schneider2014convex}. Any subfamily of an
  essential family is essential, so
  $\{\Conv(A_1),\ldots, \Conv(A_r)\}$ is essential.  Finally, we
  notice that $\AffineSpan(\Conv(A_i))=\AffineSpan(A_i)$ so the family
  $A_1,\ldots, A_r$ is essential.
\end{proof}

\begin{lemma}\label{lemma:combinatorial_criterion}
  Let $r\leq n$ and $d_1,\ldots, d_r\in\N, A_1,\ldots, A_r, E_F$ be as
  in \Cref{thm:polytopalcomplinter}. If for all
  $F\in\Faces(P)\cup \{P\}$ either $E_F$ is not essential or
  $\card{E_F}\geq \dim(F)+r-n$, then
  for all $F\in\Faces(P)\cup \{P\}$,
  $\card{E_F}\geq \dim(F)+r-n$.
\end{lemma}

\begin{proof}
We prove the statement by contraposition. Let $F\in \Faces(P)\cup \{P\}$ be
  such that $\card{E_F}< \dim(F)+r-n$. Our goal is to prove that there exists $F'\in
  \Faces(P)\cup \{P\}$ such that $\card{E_{F'}}< \dim(F')+r-n$ and $E_{F'}$ is essential.
Let $F'\in\Faces(F)\cap\{F\}$ be a face which is minimal among those
satisfying $\card{E_{F'}}< \dim(F')+r-n$. By assumption, such a face
exists. We also have that $\card{E_{F'}} - \dim(F') < r-n$.
Therefore, by minimality, for all $\widetilde F\in\Faces(F')$,
$\card{E_{\widetilde F}}\geq \dim(\widetilde F)+r-n$, and so
$\card{E_{\widetilde F}} > \dim(\widetilde F) + \card{E_{F'}} -
\dim(F')$.
Let $\AffineSpan_\R(F')$ be the smallest affine space in $M_\R$ containing $F'$.
Hence, $\dim(\AffineSpan_\R(F')) = \dim(F')$ and, as $r \leq n$,
$\card{E_{F'}} < \dim(F') = \dim(\AffineSpan_\R(F'))$.
Therefore, by \Cref{lemma:polytope_essential}, 
$E_{F'}$ is essential.
\end{proof}

\begin{proof}[Proof of \Cref{thm:regSeqInPolytopalAlg}]
  As $Y$ is defined by the intersection of $r \leq n$
  $\Q$-ample divisors, we observe that $Y$ cannot be empty and $\dim(Y)
  \geq n - r$; see the proof of \Cref{thm:Qample_subsets}.
    Since $\C[P_{M'}]$ is Cohen-Macaulay
    \cite[Thm.~1]{hochster1972rings}, by localizing at the unique homogeneous maximal
    ideal $\oplus_{i>0}\C[P_{M'}]$ and using
    \cite[Prop.~1.5.15 and Thm.~2.1.2.(c)]{bruns_cohen-macaulay_1998}, $f_1,\ldots,f_r$
    is regular if and only if $Y$ has codimension $r$.
    Hence, using the isomorphism in \cite[Thm.~7.1.13]{cox2011toric} and \Cref{thm:dimensionGenericToricSystem}, this happens if
    and only if for every face $F \in \Faces(P)\cup\{P\}$, either the
    family
    $E_F$ is not essential, or
    $\card{E_F}\geq \dim(F)+r-n$.
    By \Cref{lemma:combinatorial_criterion}, this last condition is equivalent to the fact that 
    for every face $F \in \Faces(P)\cup\{P\}$,
    $\card{E_F}\geq \dim(F)+r-n$. 
  \end{proof}

  \begin{remark*}
    We would like to point out that \Cref{thm:regSeqInPolytopalAlg}
    could also be proved algebraically along the following lines: over
    a $\N$-graded Noetherian commutative ring $R$, a homogeneous
    sequence $f_1,\ldots, f_r$ is regular if and only if for all
    $I\subseteq \conSet{1, r}$ the closed subscheme of $\Proj(R)$
    associated to the ideal $\langle f_i\rangle_{i\in I}$ has
    codimension $\card{I}$.
    This last condition can be characterised using
    \Cref{lem:genericallyGoodSubsys}, as we are dealing with
    $\Q$-ample divisors and so the aforementioned subschemes
    are not empty.
  \end{remark*}

  \subsection{Weighted homogeneous systems}

  In this subsection we study the consequences of
  \Cref{thm:regSeqInPolytopalAlg} in the context of classical
  homogeneous systems and weighted homogeneous ones. In this section, we will
  work with an affine lattice $M$, as this does not cause any problem in the
  context of polytopal algebras. We start our
  discussion with the former case.

\begin{corollary}
  Let $r\leq n$, $d_1,\ldots, d_r>0$ be integers, and $A_1,\ldots,
  A_r\subset\C[X_0,\ldots, X_n]$ be
  subsets of monomials of respective degrees $d_1,\ldots, d_r$, i.e.
  $A_i\subset \C[X_0,\ldots,X_n]_{d_i}$. Then an extremal-generic homogeneous
  system $f_1,\ldots,f_r\in\C[X_0,\ldots,X_n]$ with respective monomial
  supports $A_1,\ldots, A_r$ is regular if and only if for every subset
  $I\subset\conSet{0,n}$, 
  \[\card{\{j\in\conSet{1,r} : A_j\cap \C[X_i : i\in
  I]\ne\emptyset \}}\geq \card{I}-1 + r-n.\]
\end{corollary}

\begin{proof}
Consider the affine hyperplane
$V =\{(v_0,\ldots, v_n)\in\R^{n+1}\mid \sum_{0\leq i\leq n} v_i =
1\}\subset\R^{n+1}$ and let $M = V\cap \mathbb Z^{n+1}$ be the associated
affine lattice.
  Let $\Delta\subset M_\R$ denote the $n$-dimensional simplex which is the
  convex hull of the unit vectors in $\R^{n+1}$. 
  Then lattice points in faces of $d\cdot\Delta$ (for some $d>0$) correspond to
  homogeneous monomials of degree $d$ with respect to subsets of variables in $\C[X_0,\ldots, X_n]$.
  Applying \Cref{thm:regSeqInPolytopalAlg} with $M=V\cap \mathbb Z^{n+1}$ and $P_M
  = \Delta$ proves the corollary.
\end{proof}

Another interesting application of \Cref{thm:regSeqInPolytopalAlg} is
the context of weighted homogeneous systems: it provides us with a
characterization of the degrees for which there exist weighted
homogeneous regular sequences.
Given $\mathbf a = (a_0,\dots,a_n) \in (\mathbb Z_{> 0}^{n+1})$, 
let $\PP(a_0,\dots,a_n)$ denote the weighted projective space with
weights $(a_0,\dots,a_n)$. It is known that up to isomorphism we can
assume without loss of generality that, for every $i \in \conSet{0,n}$, the
greatest common divisor of $a_0,\dots,a_{i-1},a_{i+1},\dots,a_n$ is
$1$ \cite{delorme1975espaces}.
The weighted projective space is a normal toric variety.  To construct
its corresponding fan, consider the affine hyperplane
$V =\{(v_0,\ldots, v_n)\in\R^{n+1}\mid \sum_{0\leq i\leq n} a_i v_i =
1\}$ and set $M = V\cap \mathbb Z^{n+1}$. Then $M$ is an affine
$n$-dimensional lattice and $M\otimes \R =V$.
The dual of $M$ is
$N=\Z^{n+1}/(a_0,\ldots, a_n)\Z$. Set
$\Delta^{\mathbf a} := \{(v_0,\ldots, v_n)\in V\mid \forall i\in
\{0,\ldots, n\}, v_i\geq 0\}$. Then $\Delta^{\mathbf a}$ is a simplex,
and the data of its normal fan $\Sigma$ with respect to $N$ defines
$\PP(a_0,\ldots, a_n)$. Notice that the vertices of
$\Delta^{\mathbf a}$ need not be lattice points. A useful observation
is that $\C[\Delta^{\mathbf a}_M]$ is isomorphic (as a graded algebra)
to $\C[X_0,\ldots, X_n]$ equipped with the weighted grading
$\deg(X_0^{c_0}\cdots X_n^{c_n})= \sum_i a_i c_i$.
By \cite[Exercise 4.1.5]{cox2011toric}, the divisor class group
$\Cl(\PP(a_0,\dots,a_n))$ is isomorphic to $\Z$, and the class of a
divisor $\sum_{i=0}^n \lambda_i D_i$ (where $D_i$ is the
torus-invariant divisor associated to the ray
$\mathbb R_{\geq 0} \overline{e_i}\in N\otimes \R$, where $e_i$ is the $i$-th
canonical vector of $\Z^{n+1}$) is $\sum_{i=0}^n \lambda_i a_i$. For
$d \in \N$, the $d$-th slice $S_d$ of the Cox ring is generated by the monomials
$\{X^v : v \in \N^{n+1}, \sum_i a_i \, v_i = d\}$.

In what follows, we use \Cref{thm:regSeqInPolytopalAlg} to derive a
criterion that determines when there are regular sequences of degrees
$d_1,\dots,d_r$. By doing so, we solve an open question
asked in \cite[Remark before Sec.~3]{faugere_complexity_2016}.

\begin{theorem}\label{thm:weighthomregseq}
  Let $(a_0,\ldots, a_n)$ be positive weights, $r\leq n$ and $(d_1,\ldots,
  d_r)$ be positive integers.
  Then an extremal-generic sequence
  $f_1,\ldots, f_r\in\C[X_0,\ldots, X_n]$ of weighted homogeneous polynomials of respective weighted degrees
  $d_1,\ldots, d_r$ is regular if and only if for any subset $J\subset
  \conSet{0, n}$,
  the inequality $\card{\{ i\in \conSet{1, r} : d_i\in\sum_{j\in J}a_j\,\Z_{\geq 0}\}}\geq
  \card{J} -1 + r-n$ holds true.
\end{theorem}
\begin{proof}
  First, we notice that faces of $\Delta^{\bm a}$ 
  correspond bijectively to subsets of $\conSet{0, n}$: for $J\subset\conSet{0,
  n}$ we define the face $$F_J:=\{(v_0,\ldots, v_n)\in \R_{\geq 0}^{n+1} :
  \sum_{j\in\conSet{0, n}} a_j v_j = 1\text{ and } v_j = 0\text{ for }j\notin
  J\}.$$ Note that $\dim(F_J) = \card{J}-1$.
  The proof follows straightforwardly from \Cref{thm:regSeqInPolytopalAlg} and the fact that
  $d_i\in\sum_{j\in J}a_j\,\Z_{\geq 0}$ if and only if the face
  $d_i\cdot F_J$ contains a lattice point.
\end{proof}

\section{Theoretical complexity}\label{sec:complexity}

In this section we prove that characterizing the degrees leading
generically to complete intersections and/or regular sequences is a
computationally hard problem in general, in contrast to what happens
for generic homogeneous and multihomogeneous systems.

In the standard
homogeneous polynomial algebra $\C[x_0,\dots,x_n]$, any sequence $r\leq n$ of
generic dense homogeneous polynomials is regular, and for $r=n$ the generic
degree equals the Bézout bound. The
situation is more complicated in the multi-homogeneous case: determining a partition of the variables which minimizes the multi-homogeneous Bézout bound is an
NP-hard question \cite{malajovich2007computing}.
Still, determining if a generic
multi-homogeneous system with respect to a given a sequence of multi-degrees defines a complete intersection in the associated
multi-projective space can be achieved
by checking whether a permanent associated to the multi-degrees vanishes
\cite{mclennan2002expected}. This can be done in polynomial time by using Barvinok's algorithm~\cite{barvinok1997computing}.

The aim of this section is to argue that we cannot expect such polynomial-time
algorithms in the general case.
We first observe that the existence of complete intersections --- or
equivalently, regular sequences --- in the weighted projective space
is an NP-hard problem. Indeed, there is a polynomial reduction from
the knapsack problem, which is a known NP-complete
problem~\cite[Sec. 16.6, Eq. (27)]{schrijver1998theory}. A version of
this problem, see \cite{aardal2002hard}, states that given
$a_0,\dots,a_n,b \in \Z_{>0}$, deciding if there exist
$x_0,\dots,x_n \in \Z_{>0}$ such that $\sum a_i x_i = b$ is
NP-complete. This problem can be restated in our context: consider the
weighted projective space $\PP(a_0,\dots,a_n)$ and let $f$ be a
generic form of degree $b$. Then the dimension of the subscheme
defined by $f$ is $n-1$ if and only if there is
$(x_0,\dots,x_n) \in \Z_{>0}$ such that $\sum a_i x_i = b$.

One could wonder if determining the existence of these regular
sequences remains a hard problem even if we are given all the
monomials of the respective degrees. We will show that this is the
case by considering the problem of deciding if, given a tuple of
monomials, we can find a (classical) regular sequence supported on
them. We will reduce the NP-complete \emph{hitting set problem}
\cite{Karp1972} to this problem.
Let $R$ be a set of subsets of $\conSet{0, n}$.  A hitting set of $R$
is a subset $q \subseteq \conSet{0, n}$ such that $q$ intersects every
set in $R$.
The hitting set problem asks, given $R$ and
$k \in \N$, if there is a hitting set $q$ of $R$ of cardinality at most $k$.
Notice that we can assume without loss of generality that all subsets in $R$
are nonempty subsets of $\conSet{0, n}$.

\begin{proposition}
  Let $R$ be a set of nonempty subsets of $\conSet{0, n}$ and
  let $\mathcal A\subset\C[X_0,\ldots, X_n]$ be a monomial set containing for each set $r \in R$ a
  unique monomial $X^v$ of degree $n+1$ such that $v_i = 0$ if and only if
  $i\notin r$.
  Let $Y$ be the subscheme of $\PP^n$ corresponding to a generic
  homogeneous system of $n+1$ polynomials of degree $n+1$, all with the
  same monomial
  support~$\mathcal A$.
  Then the codimension of $Y$ is the minimum over the cardinalities of the hitting sets of~$R$.
\end{proposition}

\begin{proof}
  First we observe that if $q$ is a hitting set of $R$ then the
  polynomial system vanishes on the closed subscheme of $\PP^n$ associated to $\ideal{X_i : i \in q}$. Hence, the size
  of the minimal hitting set of $R$ is an upper bound for the
  codimension of the projective variety defined by the system.
  Notice also that $\C[X_0,\ldots, X_n]$ equals the polytopal algebra $\C[P]$, where $P$ is the simplex
  $\{(v_0,\ldots, v_n)\in\R_{\geq 0}^{n+1} \mid v_0+\dots+v_n = 1\}$, whose faces correspond to subsets of $\conSet{0, n}$.
  Following the notation in \Cref{thm:polytopalcomplinter}, as all the
  $n+1$ polynomials have the same Newton polytope belonging to $\R^n$,
  it follows that, for each face $F$, either $
  \mathcal A \cap ((n+1)\cdot F)$ is empty or
  $E_F$ is not essential. Then, by \Cref{thm:polytopalcomplinter}, the codimension of $Y$ is 
  $$
  \max_{\substack{
      F\in\Faces(P) \\
  \mathcal A \cap ( (n+1)\cdot F) = \emptyset}} \left( \dim(F) \right).
  $$
  By the orbit-cone correspondence, each face $F$ is associated
  to a projective space $\PP^{\dim(F)}$ given by the vanishing of a
  subset of variables $B_F \subseteq \{x_0,\dots,x_n\}$ of size
  $\card{ B_F } = n - \dim(F)$.
  Therefore, $\mathcal A \cap ( (n+1)\cdot F) = \emptyset$ if and only if
  every monomial in $\mathcal A$ involves
  at least one variable in $B_F$.
  So, the indices in the variables $B_F$ form a hitting set of $R$ of
  size $n - \dim(F)$.
  Hence, the maximal dimension of a face $F$ such that $(n+1) \cdot F$
  does not intersect $\mathcal A$
  corresponds to the minimal cardinality of a hitting set $B_F$.
\end{proof}

\section*{Acknowledgements}
The first author was partially funded by the ERC under the European’s
Horizon 2020 research and innovation programme (grant agreement
787840) and a public grant from the Fondation Mathématique Jacques
Hadamard.

\bibliographystyle{elsarticle-num}
\bibliography{biblio}

\end{document}